\documentclass[conference, compsoc]{IEEEtran}
\IEEEoverridecommandlockouts
\usepackage{cite}
\usepackage{amsmath,amssymb,amsthm,amsfonts}
\usepackage{graphicx}
\usepackage{xcolor}
\usepackage[normalem]{ulem}
\usepackage{multirow}
\usepackage{listings}
\usepackage{algorithm, algpseudocodex}
\usepackage{pifont}
\usepackage{tablefootnote}
\usepackage{subcaption}
\usepackage{enumitem}
\usepackage{url}
\usepackage{newunicodechar}
\usepackage{bm}
\usepackage[hidelinks]{hyperref}
\usepackage{cleveref}
\usepackage{xcolor}

\usepackage{multicol}
\usepackage{stmaryrd}
\usepackage{enumitem}
\usepackage{pgfplots}
\pgfplotsset{compat=1.18}
\usepackage{tikz}
\usetikzlibrary{tikzmark}
\usetikzlibrary{calc}
\usetikzlibrary{shapes}
\usetikzlibrary{positioning}
\usepackage{textcomp}
\usepackage[font=footnotesize]{caption}
\usepackage{tcolorbox}

\usepackage{soul}

\lstdefinestyle{example}{
  basicstyle=\ttfamily\footnotesize,
}

\theoremstyle{definition}
\newtheorem{definition}{Definition}[section]
\newtheorem{theorem}{Theorem}[section]
\newtheorem{corollary}{Corollary}[theorem]

\tikzstyle{arrow} = [thick,->,>=stealth]
\tikzstyle{cut} = [arrow, ultra thick, draw=blue]
\tikzstyle{source} = [rounded corners, draw=black, fill=red!30]
\tikzstyle{sink} = [source]

\newcommand{\libsodium}{Libsodium}  
\newcommand{\openssl}{OpenSSL} 
\newcommand{\tool}{\textsc{Serberus}}
\newcommand{\Tool}{\textsc{Serberus}}
\newcommand{\ourcc}{\textsc{llsct}}
\newcommand{\Ourcc}{LLSCT}
\newcommand{\toolclang}{\ourcc{}}
\newcommand{\Toolclang}{\Ourcc{}}
\newcommand{\toolpsf}{\textsc{\toolclang{}-psf}}
\newcommand{\toolssbd}{\textsc{\toolclang{}-nostl}}
\newcommand{\toolsls}{\textsc{\toolclang{}-sls}}
\newcommand{\lfence}{\textsc{lfence}}
\newcommand{\blade}{\textsc{Blade}}

\newcommand{\clang}{LLVM}

\newcommand{\hacl}{$\text{HACL}^*$}
\newcommand{\slh}{\textsc{slh}}
\newcommand{\retpoline}{\textsc{retpoline}}
\newcommand{\ipredd}{\textsc{ipredd}}
\newcommand{\ssbd}{\textsc{ssbd}}
\newcommand{\psfd}{\textsc{psfd}}
\newcommand{\lfenceall}{\textsc{f}+\textsc{retp}+\textsc{ssbd}}
\newcommand{\slhall}{\textsc{s}+\textsc{retp}+\textsc{ssbd}}
\newcommand{\pht}{PHT}
\newcommand{\btb}{BTB}
\newcommand{\bti}{\btb{}}
\newcommand{\stl}{STL}
\newcommand{\rsb}{RSB}
\newcommand{\psf}{PSF}
\newcommand{\rrsbad}{\textsc{rrsbad}}
\newcommand{\ibt}{IBT}
\newcommand{\shstk}{SHSTK}
\newcommand{\cetendbr}{\texttt{ENDBR}}

\newcommand{\ncalxmit}{\textsc{ncal-xmit}}
\newcommand{\ncalarg}{\textsc{ncal-arg}}
\newcommand{\ncascal}{\textsc{ncas-cal}}
\newcommand{\ncasctrl}{\textsc{ncas-ctrl}}

\newcommand{\ctprog}{CT}
\newcommand{\ourmodel}{ASP}
\newcommand{\Ourmodel}{\ourmodel{}}
\newcommand{\ourprimitive}{taint primitive}
\newcommand{\ourprimitives}{\ourprimitive{}s}
\newcommand{\Ourprimitive}{Taint primitive}
\newcommand{\Ourprimitives}{\Ourprimitive{}s}
\newcommand{\OurPrimitive}{Taint Primitive}
\newcommand{\OurPrimitives}{\OurPrimitive{}s}
\newcommand{\ourct}{CTS}

\newcommand{\Z}{\mathbb{Z}}

\newcommand{\callstack}{\mathit{CS}}
\newcommand{\PC}{\mathtt{PC}}
\newcommand{\SP}{\mathtt{SP}}
\newcommand{\ZR}{\mathtt{ZR}}
\newcommand{\CS}{\callstack{}}

\newcommand{\ConfSet}{\mathcal{C}}
\newcommand{\RegSet}{\mathcal{R}}
\newcommand{\InstSet}{\mathcal{I}}

\newcommand{\ValSet}{\mathcal{V}}
\newcommand{\DataMemSet}{\mathcal{M}_D}
\newcommand{\InstMemSet}{\mathcal{M}_I}
\newcommand{\ObsSet}{\mathcal{O}}
\newcommand{\LabelSet}{\mathcal{L}}
\newcommand{\LValSet}{\mathcal{V}_\LabelSet}
\newcommand{\Sec}{\textsc{sec}}
\newcommand{\Pub}{\textsc{pub}}

\newcommand{\Prog}{\mathcal{P}}
\newcommand{\InitConfSet}{\mathbf{C}_0}

\newcommand{\ObsBnz}[1]{\mathtt{bnz}\,{#1}}
\newcommand{\ObsCall}[1]{\mathtt{call}\,{#1}}
\newcommand{\ObsStore}[1]{\mathtt{st}\,{#1}}
\newcommand{\ObsLoad}[1]{\mathtt{ld}\,{#1}}
\newcommand{\ObsNone}{\varepsilon}

\newcommand{\Yields}[2]{\!\rightarrow_\mathrm{#1}^{#2}\!}

\newcommand{\RuleWF}{\text{WF}}
\newcommand{\RuleTT}{\text{TYP}}

\Crefname{figure}{Fig.}{Figs.}
\Crefname{table}{Tab.}{Tabs.}
\Crefname{appendix}{App.}{Apps.}

\newcommand{\pto}{\rightharpoonup}

\DeclareMathOperator{\Succs}{succs}
\DeclareMathOperator{\TSuccs}{tsuccs}

\newcommand{\ncalglob}{\textsc{ncal-glob}}

\newcommand{\tcfg}[1]{\to_\mathrm{tcfg}^{#1}}
\newcommand{\tcfgs}[1]{\to_\mathrm{tcfg*}^{#1}}

\newcommand{\Stack}{\mathit{DS}}
\newcommand{\fps}[1]{\mathit{PS}_{#1}}
\newcommand{\psp}[1]{\mathtt{PSP}_{#1}}
\newcommand{\ArgFn}{\mathcal{A}}

\newcommand{\rxmit}{r_\mathrm{xmit}}
\newcommand{\rncal}{r_\mathrm{ncal}}

\newcommand{\GprSet}{\RegSet_\mathrm{gpr}}

\newcommand{\ProgIn}{\Prog_\mathrm{in}}
\newcommand{\ProgOut}{\Prog_\mathrm{sct}}
\newcommand{\ProgFence}{\Prog_\mathrm{fence}}
\newcommand{\ProgFPS}{\Prog_\mathrm{fps}}

\newcommand{\Endcall}{\mathtt{ENDBR}}
\newcommand{\Load}{\mathtt{LD}}
\newcommand{\Store}{\mathtt{ST}}
\newcommand{\MyCall}{\mathtt{CALL}}
\newcommand{\Ret}{\mathtt{RET}}
\newcommand{\Bnz}{\mathtt{BNZ}}
\newcommand{\Jmp}{\mathtt{JMP}}
\newcommand{\Op}{\mathtt{OP}}
\newcommand{\Lfence}{\mathtt{LFENCE}}

\newcommand{\semantics}[1]{%
    \begin{minipage}{\linewidth}
        \footnotesize
        \begin{align*}
            #1
        \end{align*}
    \end{minipage}\vspace{1mm}
}

\makeatletter
\let\orgdescriptionlabel\descriptionlabel
\renewcommand*{\descriptionlabel}[1]{%
  \let\orglabel\label
  \let\label\@gobble
  \phantomsection
  \edef\@currentlabel{#1}%
  \let\label\orglabel
  \orgdescriptionlabel{#1}%
}
\makeatother

\Crefname{definition}{Def.}{Defs.}
\Crefname{corollary}{Cor.}{Cors.}
\Crefname{theorem}{Thm.}{Thms.}
\Crefname{equation}{Eq.}{Eqs.}

\newcommand{\StackBase}[1]{B_{#1}}
\newcommand{\StackEnd}[1]{E_{#1}}
\newcommand{\FrameSize}[1]{k_{#1}}

\newcommand{\IncPC}{\PC\texttt{++}}

\newcommand{\MarkInst}[2]{\tikzmark{#1-sw}{#2}\tikzmark{#1-se}}

\definecolor{darkgreen}{RGB}{0,100,0}
\definecolor{non-transient}{RGB}{15,101,160}
\definecolor{transient}{RGB}{209,99,68}
\definecolor{speculative}{RGB}{100,100,100}
\definecolor{operand}{RGB}{0,0,0}

\newcommand{\Transient}[1]{\textcolor{transient}{#1}}
\newcommand{\NonTransient}[1]{\textcolor{non-transient}{#1}}

\newcommand{\TokNT}{\text{\textcolor{non-transient}{seq}}}
\newcommand{\TokT}{\text{\textcolor{transient}{t}}}
\newcommand{\TransitionT}{\delta_\TokT}
\newcommand{\TransitionNT}{\delta_\TokNT}
\newcommand{\T}{\Transient{\textsc{t}}}
\newcommand{\NT}{\NonTransient{\textsc{seq}}}
\newcommand{\colornt}{\color{non-transient}}
\newcommand{\colort}{\color{transient}}
\newcommand{\colorspec}{\color{speculative}}

\newcommand{\cmark}{\ding{51}}

\newcommand{\xmark}{\ding{55}}

\newcommand{\cmarkc}{\textcolor{darkgreen}{\cmark}}
\newcommand{\xmarkc}{\textcolor{red}{\xmark}}

\newcommand{\pmarkc}{{\color{blue}{$\mathbf{\sim}$}}}

\newcommand{\PrimNCAL}{NCAL}  
\newcommand{\PrimNCAS}{NCAS}  
\newcommand{\PrimStack}{STKL} 
\newcommand{\PrimArg}{NARG}   

\newcommand{\ModelBase}{\llbracket \, \cdot \, \rrbracket}
\newcommand{\Lkg}[1]{\ModelBase_\mathrm{#1}}
\newcommand{\LkgMem}{\Lkg{mem}}
\newcommand{\LkgCT}{\Lkg{ct}}
\newcommand{\LkgArch}{\Lkg{arch}}

\DeclareRobustCommand{\backspace}{\tikz[baseline]{%
    \def\LineWidth{0.1em}%
    \def\BoxHeight{0.6em}%
    \def\BoxWidth{\BoxHeight}%
    \def\AngleDist{\BoxWidth/2}%
    \def\Inset{\BoxHeight/5}%
    \def\XShift{-\Inset/2}%
    \draw[line width=\LineWidth, anchor=base] (0,0) -- (\BoxWidth, 0) -- (\BoxWidth, \BoxHeight) -- (0, \BoxHeight) -- (-\AngleDist, \BoxHeight/2) -- cycle; 
    \draw[line width=\LineWidth, anchor=base] (\Inset+\XShift,\Inset) -- (\BoxWidth-\Inset+\XShift,\BoxHeight-\Inset);%
    \draw[line width=\LineWidth, anchor=base] (\Inset+\XShift,\BoxHeight-\Inset) -- (\BoxWidth-\Inset+\XShift,\Inset);%
}}

\DeclareRobustCommand{\halfcircle}{\tikz[baseline]{%
    \def\LineWidth{0.1em}%
    \def\Radius{0.35em}%
    \draw[line width=\LineWidth, anchor=base] (0, \Radius) circle [radius=\Radius];%
    \filldraw[line width=\LineWidth, anchor=base] (0, 0) arc [start angle=270, end angle=90, radius=\Radius] -- cycle;%
}}

\DeclareRobustCommand{\fullcircle}{\tikz[baseline]{%
    \def\LineWidth{0.1em}%
    \def\Radius{0.35em}%
    \filldraw[line width=\LineWidth, anchor=base] (0, \Radius) circle [radius=\Radius];%
}}

\DeclareRobustCommand{\EMitigation}{\textcolor{darkgreen}{\backspace}}
\newcommand{\PMitigation}{\raisebox{-0.3ex}{\color{gray}\Large\textbackspace ◐}}
\renewcommand{\PMitigation}{\raisebox{-0.1em}{\textcolor{gray}{\halfcircle}}}
\newcommand{\CMitigation}{\raisebox{-0.3ex}{\color{darkgreen}\Large\textbackspace ●}}
\renewcommand{\CMitigation}{\raisebox{-0.1em}{\textcolor{darkgreen}{\fullcircle}}}
\newcommand{\XMitigation}{-}
\newcommand{\IMitigation}{\raisebox{-0.7ex}{\color{red}\Large\textbackspace ⇧}}
\renewcommand{\IMitigation}{\textcolor{red}{$\uparrow$}}
\newcommand{\EIMitigation}{\raisebox{-0.7ex}{\color{gray}\Large\textbackspace ⌫}}
\renewcommand{\EIMitigation}{\textcolor{gray}{\backspace}}

\newcommand{\Cutset}{C}

\newcommand{\graydashedbox}[1]{%
  \tikz[baseline=(X.base)]\node [draw=gray, dashed, inner sep=2pt, rectangle] (X) {#1};%
}
\renewcommand{\graydashedbox}[1]{#1}

\newcommand{\SecTyp}{{\tau}}
\newcommand{\STGlob}{{\SecTyp_\mathrm{glob}}}
\newcommand{\STStk}[1]{{\SecTyp_\mathrm{stk}^{#1}}}
\newcommand{\STReg}[1]{{\SecTyp_\mathrm{reg}^{#1}}}

\newcommand{\hwnone}{\textsc{hwnone}}
\newcommand{\hwssbd}{\textsc{ssbd}}
\newcommand{\hwasp}{\textsc{asp}}
\newcommand{\hwasppsf}{\textsc{asp+psf}}
\newcommand{\hwaspssbd}{\textsc{asp-nostl}}

\newcommand{\update}{\!\gets\!}

\newcommand{\DynDep}[1]{\to_\mathrm{dep}^{#1}}
\newcommand{\StcDep}[1]{\to_\mathrm{dep}^{#1}}
\newcommand{\DynDeps}[1]{\to_\mathrm{dep*}^{#1}}
\newcommand{\StcDeps}[1]{\to_\mathrm{dep*}^{#1}}

\newcommand{\mycomment}[1]{}

\begin{document}

\title{Serberus: Protecting Cryptographic Code from Spectres at Compile-Time \\
{\it \small Authors' version; to appear in the Proceedings of the IEEE Symposium on Security and Privacy (S\&P) 2024} \vspace{-0.5cm}}
\author{\IEEEauthorblockN{Nicholas Mosier}
\IEEEauthorblockA{
nmosier@stanford.edu \\
Stanford University \\
Stanford, California, USA}
\and
\IEEEauthorblockN{Hamed Nemati}
\IEEEauthorblockA{%
hamed.nemati@cispa.de \\
CISPA Helmholtz Center \\
for Information Security \\
Saarbrücken, Germany}
\and
\IEEEauthorblockN{John C. Mitchell}
\IEEEauthorblockA{
jcm@stanford.edu \\
Stanford University \\ 
Stanford, California, USA}
\and
\IEEEauthorblockN{Caroline Trippel}
\IEEEauthorblockA{
trippel@stanford.edu \\
Stanford University \\ 
Stanford, California, USA}
}

\maketitle

\begin{abstract}

We present \tool{}, \textit{the first comprehensive mitigation} for hardening constant-time (\ctprog{}) code against Spectre attacks (involving the 
\pht{}, \btb{}, \rsb{}, \stl{}, and/or \psf{}
speculation primitives) on \textit{existing hardware}.
\tool{} is based on three insights. First, some hardware control-flow integrity (CFI) protections restrict transient control-flow to the extent that it may be comprehensively considered by software analyses. 
Second,  conformance to the accepted CT code discipline permits two code patterns that are unsafe in the post-Spectre era.
Third, once these code patterns are addressed, all Spectre leakage of secrets in CT programs can be attributed to one of four classes of \textit{taint primitives}---instructions that can transiently assign a secret value to a publicly-typed register. 
We evaluate \tool{} on cryptographic primitives in the \textsc{OpenSSL}, \textsc{Libsodium}, and \textsc{Hacl*} libraries.
\tool{} introduces 21.3\% runtime overhead on average, compared to 24.9\% for the next closest state-of-the-art software mitigation, which is less secure.
\end{abstract}


\section{Introduction}
\label{sec:intro}

The \textit{constant-time (CT) programming} approach~\cite{Bernstein:curve25519, Bernstein:Poly1305-AES, Molnar:pc-security-model, Fisch:Iron, Shaon:SGX-BigMatrix, Zheng:Opaque, Eskandarian:OliDB, Mishra:Oblix, Tople:oblivious-execution, Sasy2018ZeroTraceO, Ahmad2018OBLIVIATEAD, Coppens:timing-channel-x86, FactLanguage, synthct} was designed to support the safe execution of secret-processing programs, like cryptographic code~\cite{openssl,libsodium,hacl}, in the face of hardware side-channel attacks. Concretely, \ctprog{} programming requires that only \textit{safe} instructions, which create operand-independent hardware resource usage, process secrets.
\textit{Unsafe} instructions, i.e., information \textit{transmitters} (or simply \textit{transmitters})~\cite{jiyong:stt}, typically include control-flow, memory, and variable-time (e.g., floating-point~\cite{andrysco:subnormal} or 
integer division~\cite{Coppens:timing-channel-x86}) instructions.

Unfortunately, common hardware optimizations enable \textit{transient execution}\footnote{Transient execution refers to the execution of instructions that are never architecturally committed~\cite{canella2019systematic}.} to \textit{steer} secrets towards the operands of (transient) transmitters, circumventing \ctprog{} protections.
In \textit{Spectre attacks} specifically (our focus, \S\ref{sec:background-spectre}), transient execution results from control- or data-flow \textit{mispredictions} at runtime~\cite{canella2019systematic}.
On modern processors, there are five
well-documented sources of such (mis)predictions, called \textit{speculation primitives}~\cite{microsoft:taxonomy}:
conditional branch prediction (\pht{})~\cite{spectre}, 
indirect branch prediction (\btb{})~\cite{spectre},
return address prediction (\rsb{})~\cite{spectrersb},
store-to-load forwarding prediction (\stl{})~\cite{spectrev4}, 
and predictive store forwarding (\psf{})~\cite{cauligi:sok-spectre-sw-defense, guanciale:inspectre, cat-spectre}.\footnote{Abbreviations are borrowed from recent work which surveys Spectre attacks and software defenses~\cite{canella2019systematic, cauligi:sok-spectre-sw-defense}.}

\begin{table}[t]
    \centering
    \footnotesize
    \setlength\tabcolsep{1.5pt}
    \begin{tabular}{|c|l|c|c|c|c|c|c|}\hline
        {\bf Mitigation}                             & {\bf Leakage}                 & {\bf Proof}         & \pht{}       & \btb{}       & \rsb{}       & \stl{}       & \psf{}        \\\hline
        \textsc{intel-lfence}~\cite{intel:lfence}    & \quad\,\, -                   & -                   & \EMitigation & \XMitigation & \XMitigation & \XMitigation & \XMitigation  \\\hline
        \textsc{llvm-slh}~\cite{slh}                 & \quad \color{gray} $\LkgArch$ & \color{gray} \xmark & \PMitigation & \XMitigation & \XMitigation & \XMitigation & \XMitigation  \\\hline
        \retpoline{}~\cite{retpoline}                & \quad\,\, -                   & -                   & \XMitigation & \EMitigation & \IMitigation & \XMitigation & \XMitigation  \\\hline
        \ipredd{}~\cite{intel:bhi}                   & \quad\,\, -                   & -                   & \XMitigation & \EMitigation & \XMitigation & \XMitigation & \XMitigation  \\\hline
        \ssbd{}~\cite{intel:ssb}                     & \quad\,\, -                   & -                   & \XMitigation & \XMitigation & \XMitigation & \EMitigation & \EIMitigation \\\hline
        \psfd{}~\cite{intel:psfd}                    & \quad\,\, -                   & -                   & \XMitigation & \XMitigation & \XMitigation & \XMitigation & \EMitigation  \\\hline
        \lfenceall{}                                 & \quad\,\, -                   & -                   & \EMitigation & \EMitigation & \XMitigation & \EMitigation & \EIMitigation \\\hline
        \slhall{}                                    & \quad \color{gray} $\LkgArch$ & \color{gray} \xmark & \PMitigation & \EMitigation & \XMitigation & \EMitigation & \EIMitigation \\\hline
        \blade{}~\cite{vassena:blade}                & \quad $\LkgCT$                & \cmark              & \CMitigation & \XMitigation & \XMitigation & \XMitigation & \XMitigation  \\\hline
        \textsc{UltimateSLH}~\cite{uslh} & \quad $\LkgCT$ & \color{gray} \xmark & \CMitigation & - & - & - & - \\\hline
        \textsc{Swivel-CET}~\cite{Narayan:swivel}    & \quad \color{gray} $\LkgMem$  & \color{gray} \xmark & \CMitigation & \CMitigation & \CMitigation & \EMitigation & \EIMitigation \\\hline
        \tool{} (ours)                               & \quad $\LkgCT$                & \cmark              & \CMitigation & \CMitigation & \CMitigation & \CMitigation & \EMitigation  \\\hline
    \end{tabular}
    \caption{
    \tool{} versus mitigations for existing hardware. 
    \textit{Leakage:} \tool{} targets the \ctprog{} leakage model $\LkgCT$; \ctprog{} code is insecure under $\color{gray}\LkgMem$ and $\color{gray}\LkgArch$~\cite{cauligi:sok-spectre-sw-defense}.
    \textit{Speculation primitives:} 
    \CMitigation{} complete mitigation without disabling speculation; 
    \EMitigation{} disables speculation primitive (\EIMitigation~for implicitly disabled);
    \PMitigation{} incomplete mitigation;
    \IMitigation{} creates opportunities for speculation primitive to introduce transient execution;
    \XMitigation{} no mitigation.
    \tool{} is the only defense for the CT leakage model to mitigate leakage due to \textit{all} speculation primitives.
    }
    \label{tab:mitigation-comparison}
\end{table}

Hardening \ctprog{} programs against \textit{all} Spectre attacks involving \textit{any combination} of the aforementioned speculation primitives is hard.
Suitable hardware mitigations have been proposed~\cite{nda, choudhary:spt, yu:oisa, Schwarz:ConTExT}, but they require complex design changes that limit their adoption.
Several mitigations target existing hardware (\Cref{tab:mitigation-comparison}). 
However, none, nor any combination, is  suitable for securing \ctprog{} code.

\paragraph{\bf This Paper} We present \tool{},\footnote{\tool{} is named after Cerberus, a three-headed dog (representing its three mitigation passes) of Greek mythology with a serpentine tail (our hardware model \ourmodel{}) guarding the gates of Hades to prevent the dead (transient executions) from escaping (exfiltrating secrets) to the overworld (via transmitters).
}
\textit{the first comprehensive mitigation} for hardening \ctprog{}  code against Spectre attacks involving any combination of the \pht{}, \btb{}, \rsb{}, \stl{}, and \psf{} speculation primitives
on \textit{existing hardware}.
It mitigates the first four primitives in software and disables
only the last primitive (\psf{}) in hardware due to a clear 
\\ \vspace{-2.93\baselineskip} \\
performance advantage.
An implementation of \tool{} 
is readily deployable in our compiler \toolclang{},\footnote{\toolclang{} is open source and available at \url{https://github.com/nmosier/llsct}.}
which produces more performant (on average) and more secure binaries for cryptographic routines than state-of-the-art mitigations.
%
%
\tool{} is based on three key insights.


\textit{Insight 1: Hardware model.}
Mitigating Spectre attacks in software 
is challenged by (i) the impracticality of reasoning about unconstrained transient control-flow~\cite{binsec-haunted} and (ii) the overhead of managing unconstrained transient data-flow~\cite{ctfoundations:pldi:20}.
Like prior work~\cite{Narayan:swivel}, we observe that some lightweight (negligible overhead) hardware \textit{control-flow integrity (CFI)}  protections, like Intel CET~\cite{intel-cet-security-analysis},
constrain transient control-flow to the extent that it may be comprehensively considered by software analyses.
Thus, \tool{} requires that such CFI protections are enabled in the target hardware.
Moreover, we show \textit{for the first time} that while Spectre-\stl{}\footnote{Spectre-\stl{} denotes Spectre leakage enabled by the \stl{} speculation primitive.} can be efficiently mitigated in software, Spectre-\psf{} cannot.
Thus, \tool{} requires the \psfd{} speculation control, available on Intel~\cite{intel:psfd} and AMD~\cite{amd-predictive-store-forwarding} processors, to disable \psf{}.

\textit{Insight 2: Programming contract.}
We focus on hardening \ctprog{} programs against Spectre attacks~\cite{vassena:blade,ctfoundations:pldi:20,choudhary:spt}, since they already avoid leaking secrets during \textit{sequential} (i.e., non-transient) execution.
However, we observe that \ctprog{} programs, as defined previously~\cite{verifying-ct}, may contain two code patterns that are unsafe in the post-Spectre era and limit mitigated program performance: (i) transmitters with secret operands that never execute sequentially (e.g., ``\texttt{if (0) leak(secret)}''), and (ii) procedure calls or returns with secret arguments, which can often leak via \btb{} or \rsb{} misprediction.
Thus, we introduce a strengthening of \ctprog{} programming, called \textit{static constant-time (\ourct{})}, which requires that (i) all program variables have a static security type, and (ii) all call and return arguments are public.
\tool{} requires a \ourct{} program as input.

\textit{Insight 3: Taint primitives.}
We find that \textit{all} Spectre leakage in \ourct{} programs can be attributed to \textit{\ourprimitives{}}, or instructions that transiently assign a secret value to a publicly-typed register. 
We prove that there are exactly \textit{four} classes \ourprimitives{} (\Cref{fig:primitives}) for our hardware model (\S\ref{sec:hardware-model}), which assumes 
(i) Intel's CET protections, (ii) Intel's
\rrsbad{} and \psfd{} speculation controls~\cite{intel:msrs}, (iii) Intel's DOIT Mode~\cite{intel:doit}, and (iv) the \pht{}, \btb{}, \rsb{}, and \stl{} speculation primitives.
This result motivates \tool{}'s design and its \textit{correctness proof}: it consists of three compiler passes, each of which mitigates Spectre leakage due to one or more classes of \ourprimitives{}.

\begin{figure}[t]

\end{figure}


\noindent \textbf{Contributions.}
Our contributions are as follows:

\begin{itemize}[leftmargin=*]
    \item \textbf{\ourmodel{} operational semantics (\S\ref{sec:model}):} We define an operational semantics, called \ourmodel{}, that
    encodes all (sequential and transient) control- and data-flow that a program may exhibit when it runs on a microarchitecture satisfying our hardware model (\S\ref{sec:hardware-model}). 
    \ourmodel{} avoids explicitly modeling the hardware structures that give rise to speculation primitives, 
    capturing a wider range of implementations than prior work~\cite{speculative-execution-combinations, ctfoundations:pldi:20}.
    
    \item \textbf{Static Constant-Time programming (\S\ref{sec:properties}):}
    We propose \ourct{} programming, a strengthening of CT for the post-Spectre era.
    \ourct{} forbids two legal \ctprog{} code patterns (\S\ref{sec:properties:ct-limitations}) that 
    preclude efficient Spectre mitigation in software.
    %
    
    \item \textbf{\Ourprimitives{} (\S\ref{sec:properties}):}
    We define and characterize \textit{\ourprimitives{}}, instructions that cause a security type violation (\S\ref{sec:security-type-violations}) when transiently executed. Using \ourmodel{}, we prove that \ourprimitives{} are \textit{necessary} for Spectre leakage of secrets in \ourct{} programs and that \ourct{} programs exhibit exactly four classes of \ourprimitives{}.

    \item \textbf{\tool{} mitigation
    (\S\ref{sec:tool}):}
    We propose \tool{}, the \textit{first comprehensive mitigation} for existing hardware that hardens \ctprog{} code (satisfying \ourct{}) against Spectre attacks that exploit
    \pht{}, \btb{}, \rsb{}, \stl{}, and/or \psf{}.
    \tool{} offers the \textit{first software mitigation for \stl{}} that does not simply disable it (\Cref{tab:mitigation-comparison}); only \psf{} is the only speculation primitive it disables.
    Using \ourmodel{}, we prove that \tool{} mitigates all Spectre leakage of secrets in \ourct{} programs.

    \item \textbf{\ourcc{} compiler (\S\ref{sec:case-study}):}
    We build \ourcc{}, a custom fork of \clang{} 14.0.4 that implements \tool{}.

    \item \textbf{Case study (\S\ref{sec:case-study}--\S\ref{sec:results}):} 
    We evaluate \toolclang{} alongside two compositions of state-of-the-art mitigations~\cite{intel:lfence, slh, retpoline, intel:ssb}, \lfenceall{} and \slhall{}, on cryptographic primitives in the \openssl{}, \libsodium{}, and \hacl{} libraries. 
    \toolclang{} produces more performant code than the state-of-the-art on most benchmarks, introducing \textit{less than 8\% average overhead} on large-buffer benchmarks.
\end{itemize}

\section{Background}

\label{sec:background}
\subsection{Hardware Side-Channel Attacks}\label{sec:ct}
Hardware side-channel attacks involve a victim program (the \textit{sender}) running on vulnerable hardware that leaks secrets to an attacker (the \textit{receiver}).
Such hardware contains \textit{unsafe} instructions, or \textit{transmitters}~\cite{jiyong:stt}, whose execution creates operand-dependent hardware resource usage.
A receiver infers the value(s) of a transmitter's \textit{sensitive} operand(s) by
measuring hardware resource usage through various means, such as execution time~\cite{bernstein_aes_attack,exectimeexploit} and more~\cite{perf_counter_sc, guanciale:oakland16, Osvik:prime+probe+l1d,Yarom:flush+reload+llc13,cache_bleed,yan:directories, branchpred_sc,
Evtyushkin:BranchScope,andrysco:subnormal,Mult_leaky,portsmash,memjam,xu:controlledchannelattacks,
leaky_cauldron,tlb_bleed,drama,pandora:isca:21, powerexploit, acousticexploit, radiationexploit,KocherDPA,sok_em_side_channels}.

Constant-time (\ctprog{}) programming is the prevailing approach for countering hardware side-channel leakage in software~\cite{Bernstein:curve25519, Bernstein:Poly1305-AES, Molnar:pc-security-model, Fisch:Iron, Shaon:SGX-BigMatrix, Zheng:Opaque, Eskandarian:OliDB, Mishra:Oblix, Tople:oblivious-execution, Sasy2018ZeroTraceO, Ahmad2018OBLIVIATEAD, Coppens:timing-channel-x86, FactLanguage, synthct, ct-wasm}. In short, \ctprog{} programs do not supply secrets to sensitive transmitter operands in \textit{any} sequential execution.
\ctprog{} has been widely adopted in the context of cryptographic code~\cite{openssl, libsodium, hacl, mbedtls}, which must process (and not leak) secrets. Thus, a variety of tools and techniques have been proposed to support \ctprog{} code generation~\cite{FactLanguage, synthct, ct-wasm,evercrypt} and verification~\cite{verifying-ct, daniel2020binsec, barthe2019formal,stealthmem-verif,jasmin}. Prior work surveys many such approaches~\cite{sok-ct}.




\label{sec:sct}

\subsection{Spectre Attacks}
\label{sec:background-spectre}
\textit{Transient execution attacks} circumvent \ctprog{} protections by steering secrets towards the sensitive operands of (transient) transmitters~\cite{spectre, meltdown}.
These attacks leverage two high-level mechanisms for creating transient execution at runtime on modern processors: (i) \textit{faulting} instructions, and (ii) control- or data-flow \textit{mispredictions}. 
These mechanisms give rise to \textit{Meltdown} and \textit{Spectre} attacks, respectively~\cite{canella2019systematic}.

In general, Meltdown attacks (e.g., Meltdown~\cite{meltdown}, Foreshadow~\cite{foreshadow}, LVI~\cite{lvi}, MDS~\cite{mds-fallout, mds-ridl, mds-zombieload}) can be efficiently mitigated through microcode updates or modest hardware changes~\cite{intel-meltdown-patches}.
In contrast, despite a plethora of research proposals for mitigating Spectre attacks in hardware~\cite{invisispec, safespec, dawg:kiriansky, delayonmiss, Li:cond-spec, Ainsworth:muontrap, cleanupspec, Taram:CSF, Bourgeat:m16, Schwarz:ConTExT, yu:oisa, jiyong:stt, Yu:sdo, specshield, nda, dolma}, extreme complexity limits their adoption.
Thus, we focus exclusively on mitigating \textit{Spectre} attacks on \textit{existing hardware} designs.

\subsubsection{\bf Speculation Primitives}
\label{sec:background:speculation-primitives}

Spectre attacks~\cite{spectre, spectrersb, spectrev4, guanciale:inspectre, ctfoundations:pldi:20} are characterized according to distinct sources of control- and data-flow (mis)prediction on modern processors, called \textit{speculation primitives}~\cite{microsoft:taxonomy, canella2019systematic}. 
Predictions introduce \textit{speculative execution}, which may turn out to be \textit{sequential} (when predictions are correct) or \textit{transient} (when incorrect).

\paragraph{\bf Control-Flow Prediction}
\label{sec:control-flow-prediction}
To exploit instruction-level parallelism (ILP) in sequential applications, modern processors resolve a variety of program dependencies at runtime.
Among these are \textit{control dependencies} which arise due to control-flow instructions (e.g., conditional branches,  indirect branches) whose \textit{condition} and/or \textit{target} decide the program counter (PC) of the next instruction to fetch. 

To avoid frontend stalls, processors dedicate a significant amount of circuitry to control-flow prediction, giving rise to three notable speculation primitives.
\pht{} (responsible for Spectre v1~\cite{spectre} and v1.1~\cite{spectrev1p1}) denotes conditional branch prediction, which diverts speculative execution following a \textit{conditional branch} towards one of two control-flow paths (i.e., the taken or not taken path). 
\btb{} (responsible for Spectre v2~\cite{spectre}) denotes indirect branch prediction, which diverts speculative execution following an \textit{indirect branch} towards one of many control-flow paths.
\rsb{} (responsible for SpectreRSB~\cite{spectrersb}) denotes return address prediction, which operates similarly to \btb{} except that it diverts speculative execution following a \textit{return}.
In general, unconstrained indirect branch and return address prediction may direct speculative control-flow to \textit{any} program instruction or even to the \textit{middle} of an instruction~\cite{bhattacharyya:specrop}.

\paragraph{\bf Data-Flow Prediction}
\label{sec:data-flow-prediction}
\textit{Data-flow dependencies} through memory present another barrier to exploiting ILP. Specifically, loads block the execution of younger dependent instructions until they have retrieved their data.
Two notable speculation primitives result from data-flow prediction mechanisms designed to accelerate the execution of loads.
\stl{} (responsible for Spectre v4~\cite{spectrev4}) denotes store-to-load forwarding prediction, whereby a load may forward data from an older same-address store before all prior stores have resolved their addresses. That is, a \textit{load} may speculatively read from any same-address \textit{uncommitted} store or the most recent same-address \textit{committed} store.
\psf{} (responsible for Spectre PSF~\cite{ctfoundations:pldi:20,cat-spectre,guanciale:inspectre}) denotes predictive store forwarding, whereby a load may forward data from an older store before the load or store has resolved its address. With \psf{}, a \textit{load} may speculatively read from any prior \textit{uncommitted} store.

\subsubsection{Software Mitigations for Spectre Attacks}
\label{sec:spectre-mitigations}
\Cref{tab:mitigation-comparison} summarizes state-of-the-art software mitigations\footnote{``Software mitigations'' are those that are deployable on \textit{existing hardware}.
} for Spectre attacks, which we characterize below.

\paragraph{\bf Leakage Models}
\label{sec:background:leakage-model}
One distinction among software Spectre mitigations is the \textit{leakage model} they target (\S\ref{sec:leakage-model}). 
A mitigation's leakage model captures what an attacker can observe when a victim program runs on hardware of interest.
The most prevalent leakage models in the literature~\cite{cauligi:sok-spectre-sw-defense} are the \ctprog{} leakage model ($\LkgCT$)~\cite{barthe:highassurance, ctfoundations:pldi:20, binsec-haunted, guanciale:inspectre, guarnieri:spectector, exorcisingspectrev1, vassena:blade, guarnieri:contracts} and the sandbox isolation leakage model ($\LkgArch$)~\cite{guarnieri:contracts, cheang2019formal}.
We adopt the \ctprog{} leakage model in this paper due to our focus on mitigating Spectre leakage in cryptographic code~\cite{guarnieri:contracts,cauligi:sok-spectre-sw-defense}. 
The \ctprog{} leakage model exposes the control-flow and sequence of accessed memory addresses in a program's execution as \textit{observations} to an attacker.
The sandbox isolation model additionally exposes all values loaded from memory, which is unsuitable for modeling \ctprog{} programs which need to access secrets.
A less common leakage model, $\LkgMem$, models a receiver that can only observe accessed memory addresses~\cite{guarnieri:contracts}.

\paragraph{\bf Blocking Speculation Primitives}
The Spectre mitigations in \Cref{tab:mitigation-comparison} can be further classified as coarse- or fine-grained. Coarse-grained mitigations \EMitigation{} disable culprit speculation primitive(s). We explain several examples below.

One Spectre-\pht{} mitigation proposed by Intel (\textsc{intel-lfence}) inserts an \lfence{} after every conditional branch~\cite{intel:lfence}. An \lfence{} is a \textit{serializing instruction}, which guarantees that any younger instructions will not be executed (even transiently) until all older ones commit. This mitigation is complete under $\LkgCT$, but incurs a huge ($\approx 440$\%) overhead for typical software~\cite{lfenceoverhead}.
A higher-performance Spectre-\pht{} mitigation, \textit{speculative load hardening (SLH)},
masks the addresses or return values of loads inside in a conditional branch with the branch predicate~\cite{slh}. 
LLVM offers the only well-established SLH implementation (\textsc{llvm-slh}). However, it targets $\LkgArch$ and thus is incomplete for \ctprog{} code~\cite{shivakumar:spectredeclassified, exorcisingspectrev1}.
UltimateSLH~\cite{uslh} extends LLVM's SLH implementation to the $\LkgCT$ leakage model, with a performance cost.
Spectre-\btb{} can be mitigated with \retpoline{}~\cite{retpoline} or Intel's \ipredd{} (i.e., \textsc{ipred\texttt{\_}dis\texttt{\_}u}) speculation control~\cite{intel:bhi}.
\retpoline{} replaces all indirect branches with returns that
direct mispredictions to a ``safe'' target.\footnote{We assume the victim process sets the \rrsbad{} speculation control (\S\ref{sec:hardware-model}),
which prevents attacks from circumventing \retpoline{}~\cite{retbleed}.
}
Setting \ipredd{} disables indirect branch prediction.
%

No complete nor deployable mitigations of Spectre-\rsb{} exist under $\LkgCT$.
Intel has proposed an (incomplete) technique called \textit{RSB stuffing} to ``reduce the likelihood of an [RSB] underflow from occurring''~\cite{retpoline:skylake}.

Spectre-\stl{} and Spectre-\psf{} can be mitigated with \ssbd{} and \psfd{} speculation controls~\cite{intel:msrs}, respectively, offered by Intel~\cite{intel:msrs} and AMD~\cite{amd-predictive-store-forwarding}.

\paragraph{\bf Preventing Secret-Dependent Transmitters}
\label{sec:background:swivel}
A couple fine-grained Spectre mitigations have also appeared in the literature. These approaches \CMitigation{} explicitly prevent secrets from being supplied to transmitters during transient execution, while leaving speculation primitives enabled.

Blade~\cite{vassena:blade} is a complete Spectre-\pht{} mitigation for \ctprog{} WebAssembly under $\LkgCT$ that frames \lfence{} insertion as a min-cut data-flow problem.
Swivel-CET~\cite{Narayan:swivel} uses a control-flow tracking technique called \textit{register interlocking} to mitigate Spectre attacks involving 
\pht{}, \btb{}, and/or
\rsb{} in WebAssembly programs;
however, it targets $\LkgMem$ and thus is insecure for \ctprog{} code. 

\paragraph{\bf Layering Spectre Mitigations}
State-of-the-art mitigations \slh{}, \retpoline{}, and \ssbd{} incur modest overhead when deployed in isolation to mitigate Spectre attacks involving a \textit{single} speculation primitive.
Since there are no comprehensive mitigations for Spectre under $\LkgCT$, \S\ref{sec:case-study}-\ref{sec:results} compare \tool{}'s performance to composite baselines, \lfenceall{} and \slhall{}, which combine \textsc{intel-lfence} (\textsc{f}) and \textsc{llvm-slh} (\textsc{s})
with \textsc{retpoline} and \textsc{ssbd}.


\subsection{Threat Model}
\label{sec:threat-model}

\subsubsection{Software Model}
We \textit{provably} harden trusted \ourct{} (a slight strengthening of \ctprog{}, \S\ref{sec:propreties:ourct})
%
victim code against Spectre attacks that arise on processors satisfying the constraints of our hardware model (\S\ref{sec:hardware-model}). 
We assume a powerful attacker that can directly observe the sensitive operands of transmitters executed (sequentially or transiently) by the victim and fully control speculation within the victim process.

\subsubsection{Hardware Model}
\label{sec:hardware-model}
We assume the victim is running on a processor that soundly refines
our abstract speculative processor model, \ourmodel{} (\S\ref{sec:model}). That is, all (sequential and transient) control- and data-flow that a program can exhibit on the hardware is captured by \ourmodel{}.
Real-world processors that qualify include Intel Alder Lake N (and other, \S\ref{app:cet}) x86 CPUs, which feature: (i) Intel CET, (ii) 
\rrsbad{} and \textsc{psfd} speculation controls~\cite{intel:msrs}, (iii) Data Operand Independent Timing (DOIT)~\cite{intel:doit}, and (iv) the \pht{}, \btb{}, \rsb{}, and \stl{} speculation primitives.
\paragraph{\bf Indirect Branch Tracking}
Intel CET's \textit{indirect branch tracking (IBT)} feature helps \tool{} constrain \textit{forward-edge} control-flow speculation. IBT requires that
an \cetendbr{} instruction be placed at all valid indirect branch targets.
The processor raises an exception, or blocks transient execution, on an indirect jump to a non-\cetendbr{} (\S\ref{app:cet}).
\tool{} ensures \cetendbr{}s are only placed at a program's procedure entrypoints, so the set of possible predicted targets for an indirect branch is the set of procedure entrypoints.

\paragraph{\bf Shadow Stack and \rrsbad{}}
Intel CET's \textit{shadow stack (SHSTK)} and \rrsbad{} speculation control help \tool{} constrain \textit{backward-edge} control-flow speculation.
\rsb{} predictions are typically sourced from a hardware return stack buffer; (sequential or transient) calls push return addresses onto the return stack buffer which are popped off on return predictions.
SHSTK ensures that transient out-of-bounds stores to return addresses on the software stack cannot hijack speculative control flow~\cite{bhattacharyya:specrop}. 
The \rrsbad{} (i.e., \texttt{RRSBA\_DIS\_U}) speculation control~\cite{intel:msrs} prevents a processor from falling back to the indirect branch predictor to service return predictions on return stack buffer underfills~\cite{retbleed}.
Together, SHSTK and \rrsbad{} guarantee that the set of possible predicted targets for a return is exactly the set of instructions that immediately follow program calls.\footnote{Note that the Linux kernel performs \rsb{} filling on context switches to prevent cross-address-space Spectre-\rsb{} attacks~\cite{linux-spectre-side-channels}.}

\paragraph{\bf \textsc{psfd}}
We find that mitigating Spectre-\psf{} in software incurs significant performance overhead for the cryptographic workloads we evaluate (\S\ref{sec:results:variants}), while disabling \psf{} with Intel's \textsc{psfd} speculation control introduces negligible overhead. We advocate for the latter strategy in this paper.

\paragraph{\bf Data Operand Independent Timing}
Intel's DOIT Mode~\cite{intel:doit} guarantees that the instructions which \ctprog{} programming regards as safe (e.g., \texttt{ADD}, \texttt{XOR}, \texttt{MUL}) exhibit operand-independent timing.
Given the guarantees of DOIT, our \tool{} implementation assumes that the following instructions are transmitters: conditional branches, indirect branches, loads, stores, and division instructions.
However, \tool{} is \textit{parameterized} by a set of user-identified transmitters to encompass others that may emerge~\cite{pandora:isca:21,synthct}.



\section{\ourmodel{}: An Operational Semantics for an Abstract Speculative Processor}

\label{sec:model}
\label{sec:our-model}
We define a novel operational semantics, called \ourmodel{}, to support the design and prove the security of our Spectre mitigation \tool{}. 
\ourmodel{} consists of a \textit{leakage model}  and an \textit{execution model}~\cite{cauligi:sok-spectre-sw-defense}.
Its leakage model (\S\ref{sec:leakage-model}) refines the \ctprog{} leakage model.
Its execution model (\S\ref{sec:exec-model}-\S\ref{sec:model:semantics}) defines a \textit{non-deterministic abstract speculative processor} that executes assembly-style programs as a series of state transitions.
\Ourmodel{} captures the \pht{}, \btb{}, \rsb{}, and \stl{} speculation primitives.
We show in \S\ref{sec:extension} how to extend \ourmodel{} to also capture \psf{}.
However, since \tool{} assumes it is disabled (for performance, \S\ref{sec:results:variants}), we omit it from core \ourmodel{}.

\subsection{\ourmodel{} Preliminaries}
\label{sec:model:states}
\label{sec:model:preliminaries}
We first define basic syntax, \ourmodel{}'s configurations (i.e., system states), and its assembly-style programs.


\subsubsection{Storage and Labeled Data}
\label{sec:storage}
\paragraph{\bf Registers}
\ourmodel{} has a finite set of general-purpose registers $\GprSet$, a stack pointer $\SP$, a program  counter $\PC$, and a zero register $\ZR$.
We denote the set of all registers as $\RegSet = \GprSet \cup \{\SP,\PC,\ZR\}$.


\paragraph{\bf Values and security labels}\label{sec:model:labels}
\ourmodel{} computes on security-labeled data.
$\ValSet \subseteq \Z$ is the set of data values that registers and data memory locations can assume ($0 \in \ValSet$);
$\LabelSet = \{\Pub, \Sec\}$ is the set of security labels, where $\Pub$ and $\Sec$ mark a public (low) and secret (high) value, respectively.
$\LValSet = \ValSet \times \LabelSet$ denotes the set of \textit{labeled values}.
We use either subscript $v_l$ or pair notation $(v, l)$ to attach a label $l \in \LabelSet$ to value $v \in \ValSet$.
Given a function $o: \ValSet^n \to \ValSet$ over unlabeled values, we define the labeled equivalent $o_\LabelSet : \LValSet^n \to \LValSet$ as 
$o_\LabelSet((v_1, l_1), \ldots, (v_n, l_n)) = o(v_1, \ldots, v_n)_L$
where the out-label $L = \Sec$ iff \textit{any} in-label $l_i = \Sec$.
All arithmetic operations in \ourmodel{} propagate security labels in this way (\S\ref{sec:asm-other-inst}), as in prior work \cite{ctfoundations:pldi:20}.

\paragraph{\bf Memory Addresses}
We denote the set of mapped instruction and data addresses with $\InstMemSet, \DataMemSet \subseteq \ValSet$, respectively.
Defined by programs (\S\ref{sec:model:program}) when execution starts (\S\ref{sec:executions}), these maps are fixed thereafter.

\subsubsection{Configurations}    
\label{sec:model:configurations}
A configuration of \ourmodel{} is a five-tuple $C = (R, D, S, \CS, T)$, where:
\begin{itemize}[leftmargin=*]
    \item $R : \RegSet \to \LValSet$ is the labeled \textit{register file} contents.
    \item $D : \DataMemSet \to \LValSet$ is the labeled \textit{data memory} contents.
    \item $S \in (\DataMemSet \times \LValSet)^*$ is the \textit{speculative store set}, a list of (data address, labeled value) pairs.
    \footnote{$X^*$ denotes the set of all sequences of elements from set $X$.}
    \item $\CS \in \InstMemSet^*$ is the \textit{call stack}, i.e., a list of return addresses.
    \item $T \in \{\NT,\T\}$ is the \textit{transient execution bit}, which defines whether the configuration is transient ($T = \T$) or sequential ($T = \NT$). 
    Instructions that take transient transitions (\S\ref{sec:t-transition}) set $T \update \T$.
\end{itemize}

\label{sec:model:initial-configurations}
An \textit{initial configuration} has the form $C_0 = (R_0[\PC \update 0_\Pub], D_0, (), (), \NT)$.
$\ConfSet$ is the set of all configurations.

\subsubsection{Instructions, Programs, Security Policies}
\label{sec:model:program}
\ourmodel{} defines nine \textit{instructions}:
    \begin{align*}
    \InstSet = \,\, &\Jmp\,d \mid \Bnz\,r_\text{src},d \mid \MyCall\,r_\text{src} \mid \Ret \mid \Endcall \mid \Lfence \\ 
    &\Load \, [r_a+d],r_\text{dst} \mid \Store \, [r_a+d],r_\text{src} \mid \Op_o \, r_\text{dst}, \Vec{r_s}
    \end{align*} 
    where $d \in \Z$; $r_\text{src}, r_a, r_{s,i} \in \RegSet \setminus \PC$; and $r_\text{dst} \in \GprSet \cup \{\SP\}$.
   
A \textit{program} is a four-tuple $\Prog = (\InstMemSet, \DataMemSet, P, \InitConfSet)$, where $\InstMemSet$ and $\DataMemSet$ (\S\ref{sec:storage}) are the mapped instruction and data addresses, respectively, and: 
\begin{itemize}[leftmargin=*]
    \item $P : \InstMemSet \to \InstSet$ is the read-only \textit{instruction memory} contents.
    We write $I \mapsto \mathbb{I}$
    to mean $P(I) = \mathbb{I}$. 
    \item $\InitConfSet \subseteq \ConfSet$ are the initial configurations of the program. 
\end{itemize}


\noindent Each initial configuration $C_0 \in \InitConfSet$ of program $\Prog$ defines 
the initial memory and register contents, which implicitly specifies the \textit{security policy} (i.e., a labeling of all initial program values).
Our security-typeability requirement (\S\ref{sec:security-typeable}) constrains the security policy of initial configurations.

\subsection{\ourmodel{}'s Leakage Model}
\label{sec:leakage-model}
%
\paragraph{\bf Observations}
\label{sec:model:observations}
\label{sec:model:transmitters}
\ourmodel{} realizes the \ctprog{} leakage model, or $\LkgCT$ (\S\ref{sec:spectre-mitigations}), by defining a set of four transmitters. When they execute (\S\ref{sec:exec-model}), transmitters expose their sensitive operand(s) as an observation taken from \textit{observation set} $\ObsSet$:
$$\ObsSet = \, \ObsNone \mid \ObsBnz{v_l} \mid \ObsCall{v_l} \mid \ObsLoad{A_l} \mid \ObsStore{A_l}$$
Observation $\mathtt{bnz}\,\,v_l$ exposes the labeled condition of a conditional branch ($\Bnz$); $\mathtt{call}\,\,v_l$ exposes the labeled target of a call ($\MyCall$); 
$\mathtt{ld/st}\,\,A_l$ exposes the labeled address operand of a load/store ($\Load$/$\Store$). All other instructions are safe and expose the empty observation $\ObsNone$.

Return instructions ($\Ret$) are not transmitters, since they cannot leak \textit{new} secrets on processors with a SHSTK (\S\ref{sec:hardware-model}).
On such designs, a secret return address implies a secret $\PC$ at some prior call (the one that pushed it onto the return stack buffer or SHSTK).
A secret $\PC$ implies the execution of a prior $\MyCall$ or $\Bnz$ with a secret sensitive operand. Thus, the secret already leaked.



In \S\ref{sec:extension}, we describe how to extend \ourmodel{} with a fifth transmitter, division instructions ($\mathtt{DIV}$).
Without loss of generality, we omit $\mathtt{DIV}$ from core \ourmodel{}.
Our \tool{} implementation (\S\ref{sec:case-study}) \textit{does} assume $\mathtt{DIV}$s are transmitters.

%



\paragraph{\bf Declassification} \label{sec:model:declassification}
While \textit{labeled} transmitter operands are exposed as observations, \ourmodel{} declassifies transmitter operands (by assigning labels to $\Pub$ or stripping labels entirely) before using them to compute the next configuration (\S\ref{sec:model:semantics}).
Doing so eliminates \textit{secondary} (repeated) leaks of the same secret while retaining \textit{primary} (initial) leaks.
Declassification is sound (i.e., it does not result in missed Spectre leakage), since \tool{} ensures the program never passes a secret to a transmitter in the first place.

\subsection{\ourmodel{}'s Execution Model: Overview}
\label{sec:exec-model}
\ourmodel{} captures all possible speculative executions of programs running on a design that satisfies our hardware assumptions (\S\ref{sec:hardware-model}), using a
\textit{sequential semantics} ($\TransitionNT$, \S\ref{sec:nt-transition}) and a \textit{transient semantics} ($\TransitionT$, \S\ref{sec:t-transition}).

\subsubsection{Limitations of Prior Execution Models}
\label{sec:exec-model-limitations}
\label{sec:model:limitations}

\begin{table}[t]
    \centering\footnotesize
    \setlength\tabcolsep{2pt}
    \begin{tabular}{|c|c|c|c|c|c|c|c|c|}\hline
        \bf Execution model                                     & \textsc{pht} & \textsc{btb}   & \textsc{rsb}     & \textsc{stl}     & \textsc{psf}   & \textsc{ibt} & \textsc{shstk} & \textsc{rrsbad} \\\hline
        Guarnieri et al. \cite{guarnieri:spectector}            & \checkmark   &                &                  &                  &                &              &                &                 \\\hline
        Cauligi et al. \cite{ctfoundations:pldi:20}             & \checkmark   & $\checkmark^*$ & $\checkmark^{*-}$ & $\checkmark$     & \checkmark     &              &                & \checkmark      \\\hline
        Vassena et al. \cite{vassena:blade}                     & \checkmark   &                &                  &                  &                &              &                & \\\hline
        Fabian et al. \cite{speculative-execution-combinations} & \checkmark   &                & $\checkmark^{-}$ & $\checkmark^{-}$ &                &              &                & \checkmark      \\\hline
        \ourmodel{} (ours)                                      & \checkmark   & \checkmark     & $\checkmark$     & $\checkmark$     & $\checkmark^*$ & \checkmark   & \checkmark     & \checkmark      \\\hline
    \end{tabular}
    \setlength\tabcolsep{6pt}
    \caption{
        A comparison of formal speculative execution models.
        $\checkmark$ indicates the model captures that speculation primitive (\pht{}, \btb{}, \rsb{}, \stl{}, \psf{}), speculative CFI protection (\ibt{}, \shstk{}), or speculation control (\rrsbad{});
        $\checkmark^{-}$ indicates the model restricts the behavior of the speculation primitive (\S\ref{sec:model:limitations});
        $\checkmark^*$ indicates the speculation primitive is captured in an \textit{extension} of the core execution model (e.g., \S\ref{sec:extension} for \ourmodel{}).
    }
    \label{tab:model-comparison}
\end{table}



Several execution models have been proposed to support software analyses that detect or mitigate Spectre leakage in \ctprog{} programs (\Cref{tab:model-comparison})~\cite{cauligi:sok-spectre-sw-defense}.
Two limitations of these models motivate us to design something new. First, none capture the semantics of Intel CET protections (IBT and SHSTK).
Second, despite capturing various combinations of the five speculation primitives we target, these models make highly specific (and even unrealistic) assumptions about their microarchitectural implementations. 
For example, Fabian et al.~\cite{speculative-execution-combinations} model \stl{} by transiently ``deleting'' stores,
which misses Spectre leakage on realistic designs.\footnote{
Consider ``\texttt{a=1; a=0; if (a /*L1*/) b = secret; if (a /*L2*/) leak(b);}''
in which \texttt{secret} transiently leaks only if L1 skips over \texttt{a=0} (and reads from \texttt{a=1}), but L2 reads from \texttt{a=0}.
}
Cauligi et al.~\cite{ctfoundations:pldi:20} model \rsb{} using an unrealistic infinite-size stack structure and thus miss leakage due to RSB overfills.



\subsubsection{Sequential Transitions} 
\label{sec:nt-transition}
Given a configuration $C$ and instruction memory $P$, the \textit{sequential ($\TokNT$) transition} function $\TransitionNT(C, P) = (C', O)$ returns a successor configuration $C'$ and observation $O \in \ObsSet$.
We say that ``$C$ sequentially yields $C'$ exposing $O$,'' written $C \Yields{\TokNT}{O} C'$.
Each configuration has \textit{exactly one} sequential successor, as there is only one sequential execution of a program.

\subsubsection{Transient Transitions}
\label{sec:t-transition}
The \textit{transient ($\TokT$) transition} function $\TransitionT(C, P) \!\subseteq\! \ConfSet \!\times\! \ObsSet$ returns a \textit{set} of (configuration, observation) pairs,
capturing all transient execution steps that may arise on behalf of some speculation primitive (\S\ref{sec:background:speculation-primitives}). 
We say ``$C$ transiently yields $C'$ exposing $O$,'' written $C \Yields{\TokT}{O} C'$, if $(C', O) \in \TransitionT(C,P)$.

We say an instruction $I \in \InstMemSet$ is a \textit{speculation primitive} if there exists a configuration $C = (R, D, S, \CS, T)$ satisfying $R(\PC) = I$ that has both a transient transition (i.e., $\TransitionT(C, P) \neq \emptyset$) and a non-faulting sequential transition (i.e., $\TransitionNT(C, P) \neq (C, O)$, \S\ref{sec:model:accesses}).
In \ourmodel{}, speculation primitives are instances of $\Bnz$, $\MyCall$, $\Ret$, or $\Load$.
\textit{Speculative execution} encompasses both sequential and transient execution (\S\ref{sec:model:terminology}). Thus, we say $C$ (speculatively) yields $C'$ exposing $O$, written $C \Yields{}{O} C'$, if $C \Yields{\TokNT}{O} C'$ or $C \Yields{\TokT}{O} C'$.


\subsubsection{Traces}
\label{sec:executions}
\label{sec:traces}
A \textit{trace} of a program $\Prog = (\InstMemSet, \allowbreak \DataMemSet, \allowbreak P, \allowbreak \InitConfSet)$ from an initial configuration $C_0 \in \InitConfSet$ is the sequence of transitions $e = C_0 \Yields{}{O_0} \allowbreak C_1 \Yields{}{O_1} \allowbreak \cdots \allowbreak \Yields{}{O_{n-1}} C_n.$
We use subscript-$i$ notation to denote components of a configuration $C_i = (R_i, D_i, S_i, \CS_i, T_i)$ at step $i$ of a trace.
We say $I_i = R_i(\PC)$, i.e., $I_i$ is the address of the instruction executing step $i$.
%
\label{sec:model:terminology}%
A trace $e = C_0 \Yields{}{O_0} \cdots \Yields{}{O_{n-1}} C_n$ is \textit{transient} (resp. \textit{sequential}) \textit{at step $i$} if $T_i = \T$ (resp. $T_i = \NT$). 
If $i$ is not specified, assume $i = n$, i.e., the last step in a trace.
We say an \textit{instruction} $I_i$ executes transiently (resp. sequentially) if the trace is transient (resp. sequential) at step $i+1$.
\subsection{\ourmodel{}'s Instruction Execution Semantics}
\label{sec:model:semantics}
We specialize 
\ourmodel{}'s transition functions per instruction.\footnote{%
We use color to distinguish \textit{\colornt sequential} and \textit{\colort transient} semantics, so readers should view the paper in color for a better experience.}
Assume current configuration $C \!=\! (R, D, S, \CS, T)$, program $\Prog \!=\! (\InstMemSet, \DataMemSet, P, \InitConfSet)$, and current instruction address $I \!=\! R(\PC)$.
If $I$ is unmapped ($I \!\not \in\! \InstMemSet$), the processor halts (i.e., $\TransitionNT(C,P) \!=\! (C, \ObsNone)$ and $\TransitionT(C,P) \!=\! \emptyset$).



\subsubsection{Conditional Branches (PHT)}
\label{sec:model:bnz}
The $\Bnz$ instruction captures the \pht{} speculation primitive.
Its \textit{sequential} transition takes the branch by adding a fixed displacement $d$ to the $\PC$ if the value in register $r$ is nonzero; otherwise, it falls through to the next instruction.
The \textit{transient} transition does the opposite, modeling a branch misprediction.
Both transitions expose the labeled branch condition via observation $\ObsBnz{R(r)}$ and then declassify it before computing the branch target. \\
\semantics{
    \frac{
        \begin{aligned}
            &\textsc{Cond. Br.} \quad I \mapsto \Bnz\,r,d \qquad \colorspec v_l = R(r) \quad \colorspec c = (v \neq 0) \quad \colorspec O = \ObsBnz{v_l} \\
            &\colornt I'_\TokNT = I + 1 + {\colorspec c} \cdot d \quad\,\,
            \colort I'_\TokT = I + 1 + (1 \!-\!  {\colorspec c}) \cdot d \quad\,\,
            \colort R'_\TokT  = R[\PC \update (I'_\TokT)_\Pub] \\
            &\colornt R'_\TokNT = R[\PC \update (I'_\TokNT)_\Pub] \,\,
            \colort C'_\TokT = C[R \update R'_\TokT; T \update \T] \,\,
            \colornt C'_\TokNT = C[R \update R'_\TokNT]
        \end{aligned}
    }{
        \begin{aligned}
            \colornt \TransitionNT(C, P) \colornt= (C'_\TokNT, {\colorspec O}) \quad
            \colort \bm{\TransitionT(C, P)} \colort= \{(C'_\TokT, {\colorspec O})\}
        \end{aligned}
    }
}

\subsubsection{Calls (BTB, IBT)}
\label{sec:model:call}
The $\MyCall$ and $\Endcall$ instructions together model the \btb{} speculation primitive and Intel's IBT.
$\MyCall$'s \textit{sequential} transition checks whether the target instruction at $I'_\TokNT \!=\! R(r)$ is an $\Endcall$.
If so, it jumps to $I'_\TokNT$; 
if not, execution halts due to a CFI violation.
A \textit{transient} transition may jump to \textit{any} $I'_\TokT \mapsto \Endcall$ in the program ($I'_\TokT \!\neq\! I'_\TokNT$).
All transitions push the return address $I_r = I + 1$ onto the callstack $\CS$, expose the (sequential) target via observation $\ObsCall{R(r)}$, and declassify the resulting $\PC$.
\\
\semantics{
    &\textsc{Call} \quad I \mapsto \MyCall \, r \qquad \colornt (I'_\TokNT)_l = R(r) \qquad \colorspec 
    I_r = I+1 \\
    &\frac{
        \begin{aligned}
            &\colort \mathbf{I'_\TokT} = \{ I'_t \mid I'_t \mapsto \Endcall \} \setminus {\colornt I'_\TokNT}  \quad
            \colorspec \CS' = (\CS :: I_r) \quad
            O = \ObsCall{R(r)} \\
            &\colornt R'_\TokNT = R[\PC \update (I'_\TokNT)_\Pub] \quad
            \colort \bm{\mathbf{R'_\TokT}} = \{ R[\PC \update (I'_\TokT)_\Pub] \mid I'_\TokT \in \mathbf{I'_\TokT} \}
        \end{aligned}
    }{
        \begin{aligned}
            &\colornt \TransitionNT(C, P) = (C[R \update R'_\TokNT ;\, \CS \update {\colorspec \CS'}], {\colorspec O}) \\
            &\colort \bm{\TransitionT(C, P)} = \{ (C[R \update R'_\TokT ;\, \CS \update {\colorspec \CS'};\, T \update \T], {\colorspec O}) \mid R'_\TokT \in \mathbf{R'_\TokT} \}
        \end{aligned}
    }
} \\
The sole purpose of the $\Endcall$ instruction is to mark valid targets of calls; it executes as a no-op.
\\
\semantics{
    &\frac{
        \textsc{Endbr} \quad I \mapsto \Endcall \qquad  \colornt R'_\TokNT = R[\IncPC]
    }{
        \colornt \TransitionNT(C, P) = (C[R \update R'_\TokNT], \ObsNone) \qquad
        \colort \bm{\TransitionT(C, P)} = \emptyset
    }
} 

\subsubsection{\bf Returns (\rsb{}, SHSTK)}
\label{sec:model:ret}
The $\Ret$ instruction captures the RSB speculation primitive.
If the callstack is empty, the processor halts execution.
Otherwise, it pops the sequential return address $I'_\TokNT$ off the callstack.
The \textit{sequential} transition jumps to $I'_\TokNT$; the \textit{transient} transition jumps to \textit{any callsite in the program},
i.e., any instruction address $I'_\TokT$ immediately following a $\MyCall$.
\\
\semantics{
    &\frac{
        \begin{aligned}
            &\textsc{Return} \,\, I \mapsto \Ret \quad
            \colorspec (\CS' :: {\colornt I'_\TokNT}) = \CS \quad
            \colornt R'_\TokNT = R[\PC \update (I'_\TokNT)_\Pub)] \\
            &\colort \bm{\mathbf{R'_\TokT}} \!=\! \{ R[\PC \update (I'_\TokT)_\Pub)] \mid I'_\TokT-1 \mapsto \MyCall \} \setminus {\colornt R'_\TokNT} \quad
            \colorspec C' \!=\! C[\CS \gets \CS']
            \\
            &\colornt C'_\TokNT = {\colorspec C'}[R \update R'_\TokNT] \quad
            \colort \bm{\mathbf{C'_\TokT}} = \{ {\colorspec C'}[R \update R'_\TokT; T \update \T] \mid R'_\TokT \in \mathbf{R'_\TokT} \}
        \end{aligned}
    }{
        \begin{aligned}
            \colornt \TransitionNT(C, P) &\colornt= \begin{cases}
                (C, \ObsNone) &\mbox{if } \CS \mbox{ empty} \\
                (C'_\TokNT, \ObsNone) &\mbox{otherwise}
            \end{cases} \,\,
            \colort \bm{\TransitionT(C, P)} = \mathbf{C'_\TokT} \times \{\ObsNone\}
        \end{aligned}
    }
}
\\
\ourmodel{} captures a wider range of transient behaviors due to \rsb{} than prior work (\S\ref{sec:model:limitations}).
Our approach captures all return predictor implementations that can predict previously seen return addresses only (e.g., return stack buffers).

\paragraph{\bf Shadow Stack}
$\MyCall/\Ret$ semantics and the callstack $\CS$ together model Intel's SHSTK.
By maintaining the call stack outside of data memory, \ourmodel{} prevents stores from overwriting return addresses, transiently or otherwise.
When \ourmodel{} executes a $\Ret$, it ensures that \textit{sequential} control returns to the correct callsite; \textit{transient} control may return to 
another, incorrect callsite.

\subsubsection{Loads and Stores (STL)}
\label{sec:model:accesses}
A $\Load/\Store$ instruction computes its labeled effective address $A_l \!=\! R(r_a) \!+\! d_\Pub$ and exposes it via observation $\mathtt{ld}/\mathtt{st}\,{A_l}$.
Before performing the memory access, ${A_l}$ is declassified by stripping its label to produce $A$ (recall that data addresses are unlabeled, \S\ref{sec:model:preliminaries}).

\paragraph{\bf Store}
Consider a store ($\Store$).
If $A$ is unmapped, then it has one \textit{sequential} transition that halts execution (modeling a fault) and one \textit{transient} transition that executes as a no-op (modeling out-of-order execution).
If $A$ is mapped, then the store has one \textit{sequential} transition that appends the (declassified address, value) pair $(A, R(r))$ to the speculative store set and no transient transitions.
Data memory $D$ is not updated until a speculation fence executes (\S\ref{sec:model:semantics:lfence}).
\\
\semantics{
    &\frac{
        \begin{aligned}
            &\textsc{Store} \quad I \mapsto \Store \, [r_a+d],r \qquad 
            \colorspec A_l = R(r_a) + d_\Pub \quad
            \colorspec O = \ObsStore{A_l} \\
            &\colornt S'_\TokNT = (S :: ({\colorspec A}, R(r))) \qquad
            \colorspec R' = R[\IncPC]
            \\
            &\colornt C'_\TokNT = C[R \gets {\colorspec R'}; S \gets S'_\TokNT]
            \qquad
            \colort \mathbf{C'_\TokT} = \{(C[R \gets {\colorspec R'}; T \gets \T], {\colorspec O})\}
        \end{aligned}
    }{
            \colornt \TransitionNT(C, P) = \begin{cases}
                (C'_\TokNT, {\colorspec O}) \hspace{-2mm}&\mbox{if } A \in \DataMemSet \\
                (C, {\colorspec O}) &\mbox{if } A \not\in \DataMemSet
            \end{cases}
            \,\,
            \colort \bm{\TransitionT(C, P)} = \begin{cases}
                \emptyset &\mbox{if } A \in \DataMemSet \\
                \mathbf{C'_\TokT} \hspace{-2mm}&\mbox{if } A \not \in \DataMemSet
            \end{cases}
    }
}
\paragraph{\bf Load}
Consider a load ($\Load$), whose semantics captures the \stl{} speculation primitive.
If $A$ is unmapped, then it has one \textit{sequential} transition that halts execution (modeling a fault)
and one \textit{transient} transition that assigns zero to its output register (modeling recent Intel processors, including our workstation~\cite{lvi}).
If $A$ is mapped, then the load has one \textit{sequential} transition that reads from the most recent same-address store.
It may have more than one transient transition, each of which reads from a distinct same-address store in the speculative store set \textit{or} from data memory. \\
\semantics{
    &\frac{
        \begin{aligned}
            &\textsc{Load} \quad I \mapsto \Load \, [r_a+d], r \quad
            \colorspec  \left( \left(A_1, (v_1)_{l_1} \right), \ldots, \left(A_n, (v_n)_{l_n} \right) \right) = S \\
            &\colorspec A_l = R(r_a) + d_\Pub \quad
            \colornt D_\TokNT = D[{\colorspec A_1 \update (v_1)_{l_n}; \cdots; A_n \update (v_n)_{l_n}}] \\
            &\colornt v_\TokNT = D_\TokNT({\colorspec A}) \qquad
            \colort \mathbf{v_\TokT} = \{D({\colorspec A})\} \cup  
            \{ {\colorspec (v_i)_{l_i}} \mid {\colorspec A_i} = {\colorspec A} \}
            \qquad \colorspec O = \ObsLoad{A_l}
            \\
            &\colornt R'_\TokNT = R[\IncPC ;\, r \update v_\TokNT]  \qquad
            \colort \mathbf{R'_\TokT} = \{ R[\IncPC ;\, r \update v_\TokT] \mid v_\TokT \in \mathbf{v_\TokT} \}%
        \end{aligned}
    }{
        \begin{aligned}
            \colornt \TransitionNT(C, P) &\colornt = \begin{cases}
                (C[R \update R'_\TokNT], {\colorspec O})  &\mbox{if } A \in \DataMemSet \\
                (C, {\colorspec O}) &\mbox{if } A \not \in \DataMemSet
            \end{cases} \\
            \colort \bm{\TransitionT(C, P)} &\colort= \begin{cases}
                \{ (C[R \update R'_\TokT; T \update \T], {\colorspec O}) \mid R'_\TokT \in \mathbf{R'_\TokT} \} &\mbox{if } A \in \DataMemSet \\
                \{ (C[R \update R[\PC\texttt{++}; r \gets 0_\Pub]; T \update \T], {\colorspec O}) \} \hspace{-2mm}&\mbox{if } A \not \in \DataMemSet
            \end{cases}
        \end{aligned}
    }
} \\
As with \rsb{}, \ourmodel{} models \stl{} in an \textit{implementation-agnostic} manner. It captures a 
\textit{superset} of the STL behaviors captured by prior execution models (\S\ref{sec:model:limitations}). 

\paragraph{\bf \hspace{-5mm} NCA/CA accesses}
We partition memory accesses into two classes which
differ in their ability to introduce secrets into transient computation, as we show in \S\ref{sec:properties}--\ref{sec:tool}.
\begin{definition}
\label{def:constant-address}
        A memory access $I \mapsto \Load/\Store\,[r_a+d],r$ is \textit{constant-address (CA)} if $r_a \in \{\ZR, \SP\}$; otherwise, $I$ is \textit{non-constant-address (NCA)}.
        Furthermore, if $r_a = \ZR$, $I$ is a CA \textit{global access}, and displacement $d$ is the fixed address of a global variable. 
        If $r_a = \SP$, $I$ is a CA \textit{stack access}, and $d$ gives the frame offset of a stack variable.
\end{definition}

\subsubsection{Speculation Fence}\label{sec:model:semantics:lfence}
\ourmodel{} features a \textit{speculation fence} instruction, $\Lfence$, which halts transient execution but allows sequential execution to proceed.
Its semantics depends on whether the current configuration is sequential ($T=\NT$) or transient ($T=\T$).
If it is \textit{sequential}, $\Lfence$ drains all stores in the speculative store set $S$ to data memory, thereby 
restricting the set of stores that loads may transiently forward data from (\S\ref{sec:model:accesses}).
If it is \textit{transient}, execution halts.
\\
\semantics{
    &\frac{
        \begin{aligned}
            &\textsc{Speculation Fence} \quad I \mapsto \Lfence \qquad
            \colornt R'_\TokNT = R[\IncPC] \\
            &\colornt \left( \left(A_1, (v_1)_{l_1} \right), \ldots, \left(A_n, (v_n)_{l_n} \right) \right) = S \\
            &\colornt D'_\TokNT = D[A_1 \update (v_1)_{l_1} ; \cdots ; A_n \update (v_n)_{l_n}] \\
            &\colornt C'_\TokNT = C[R \update R'_\TokNT ;\, D \update D'_\TokNT ;\, S \update ()]
        \end{aligned}
    }{
        \colornt \TransitionNT(C, P) = \begin{cases}
            (C, \ObsNone) &\mbox{if } T = \T \\
            (C'_\TokNT, \ObsNone) &\mbox{otherwise}
        \end{cases} \qquad
        \colort \bm{\TransitionT(C, P)} = \emptyset
    }
}
%

\subsubsection{Other Instructions}
\label{sec:asm-other-inst}
\label{sec:model:other}
$\Jmp$ \textit{sequentially} jumps to $\PC \!+\! d \!+\! 1$.
$\Op_o$ represents a class of arithmetic operations (e.g., $\mathtt{MOV}$, $\mathtt{ADD}$), parameterized by function $o \!: \ValSet^n \!\to\! \ValSet$. $\Vec{r_s}$ is a list of input registers; $r$ is the output register.
Its \textit{sequential} transition assigns $r \update o_\LabelSet(R(r_{s,1}), \ldots, R(r_{s,n}))$.
Neither $\Jmp$ nor $\Op$ have transient transitions.
We omit transition rules here for brevity, but provide them in \S\ref{app:semantics}.

\section{Characterizing Spectre Leakage in Static Constant-Time Programs}
\label{sec:properties}

\subsection{Speculative Constant-Time}
\label{sec:properties:sct}

We formalize Spectre leakage of secrets in \ctprog{} programs
as a violation of the \textit{speculative constant-time (SCT)} security property from prior work~\cite{vassena:blade, ctfoundations:pldi:20,binsec-haunted}.
A program satisfies SCT on \ourmodel{} if there does not exist a trace that exposes secret-dependent observations.
An \textit{SCT violation} is a trace that exposes some secretly-labeled
observation $O$.
Cauligi et al.~\cite{ctfoundations:pldi:20} show that all Spectre attacks manifest as SCT violations.
\begin{definition}
\label{def:speculative-constant-time}
A program $\Prog$ is SCT iff for all traces $e = C_0 \Yields{}{O_0} \cdots \Yields{}{O_n} C_{n+1}$, no observation $O_i$ is labeled secret for any $0 \leq i \leq n$.
\end{definition}

\subsection{Limitations of Traditional \ctprog{}}
\label{sec:properties:ct-limitations}
\begin{definition}
\label{def:constant-time}
A program $\Prog$ is \textit{constant-time (CT)} iff for all \textit{sequential} traces $e = C_0 \Yields{\TokNT}{O_0} \cdots \Yields{\TokNT}{O_n} C_{n+1}$, no observation $O_i$ is labeled secret for any $0 \leq i \leq n$.
\end{definition}

We observe that two limitations of \Cref{def:constant-time} prevent efficient mitigation of SCT violations in \ctprog{} programs.

First, a \ctprog{} program may contain \textit{latent \ctprog{} violations}, such as ``\texttt{if (0) x = A[secret]},'' that do not manifest in any sequential trace of the program.
Such patterns exhibit compile-time \textit{security-type violations} by assigning a secret value (e.g., \texttt{secret}) to a public variable or supplying it to a public operand (e.g., the index operand in \texttt{A[secret]}).
Existing \ctprog{} compilers~\cite{FactLanguage, ct-wasm} detect such security-type violations during compile-time typechecking. 
We introduce a \textit{security-typeability} requirement (\S\ref{sec:security-typeable}) for \ourct{} programs, which formalizes the security-type guarantees of a \ctprog{} program that has passed such typechecks. 

%

Second, \ctprog{} programs use a \textit{Spectre-unaware calling convention} that permits
passing secrets by value during calls and returns. This is \textit{inherently unsafe} in the presence of \btb{} or \rsb{} mispredictions, which
can easily leak
secret arguments:
the call (resp. return) need only transiently jump to a procedure (resp. callsite) that expects a \textit{public} argument in a register, which it subsequently leaks by supplying it to a transmitter.
Secret argument leakage is difficult to mitigate efficiently: it forces a mitigation to conservatively assume \textit{all arguments and return values} may be transiently secret.
No combination of currently deployed mitigations for the CT leakage model can fully protect against this kind of leakage.\footnote{Even if \textsc{lfence}/\textsc{slh}, \textsc{retpoline}, and \textsc{ssbd} are simultaneously enabled, return values may \textit{still} leak via RSB speculation.}

\subsection{Static Constant-Time Programming}
\label{sec:propreties:ourct}




\begin{definition}\label{def:ourct}
    A program $\Prog$ is \textit{static constant-time (\ourct{})} iff it satisfies \ctprog{} (\Cref{def:constant-time}) as well as \textbf{WF\label{item:contract:wf}} (\textit{well-formed, \S\ref{sec:properties:ourct:wf}}) and \textbf{TYP\label{item:contract:tt}} (\textit{security-typeable, \S\ref{sec:security-typeable}}).
    
\end{definition}

We propose \textit{static constant-time (\ourct{}) programming}, a strengthening of \ctprog{} programming which overcomes the limitations in \S\ref{sec:properties:ct-limitations}.
%
We find that existing \ctprog{} programs generally satisfy \ourct{} when compiled with \texttt{-O3} optimizations and a carefully selected set of compiler flags.
\S\ref{app:unsafe-compiler} provides a full list of these flags for LLVM, which disable optimizations that are incompatible with CTS, like stack slot sharing and argument promotion. 
Using these flags, \textit{all} of the cryptographic primitives we benchmark in \S\ref{sec:case-study} satisfy \ourct{} without requiring \textit{any} source modifications.

\subsubsection{Well-Formed}
\label{sec:properties:ourct:wf}
\label{sec:well-formed}
\textit{Well-formedness} captures the important structural and behavioral properties and metadata of compiled code 
that can be leveraged by a compiler-based software mitigation like \tool{} (\S\ref{sec:tool}).

\label{sec:successors}
The \textit{intraprocedural successors} $\Succs(I) \subseteq \InstMemSet$ of an instruction address $I \in \InstMemSet$ are defined as:
$$
\Succs(I) = \begin{cases}
    \emptyset &\mbox{if } I \mapsto \Ret \\
    \{I+1,I+1+d\} &\mbox{if } I \mapsto \Bnz\,r,d \\
    \{I+1+d\} &\mbox{if } I \mapsto \Jmp\,d \\
    \{I+1\}   &\mbox{otherwise.}
\end{cases}
$$

\label{sec:properties:procedures}
    A \textit{procedure} $F \!\subseteq\! \InstMemSet$ is a set of instruction addresses that
    (i) contains exactly one $\Endcall$ instruction, denoting the entrypoint; 
    (ii) is closed under intraprocedural succession; 
    (iii) has a prologue on entry allocating a stack frame of fixed size $\FrameSize{F}$; 
    (iv) has an epilogue deallocating the frame on return; and
    (v) has only in-bounds stack accesses into its frame.
\Cref{fig:fps-example} shows an example procedure.
$F_i$ denotes the procedure containing the instruction $I_i$ executing in $i$th step of a trace.

\begin{definition}[Well-formed programs \textbf{\RuleWF{}}]
\label{def:well-formed}
     A program $\Prog$ is \textit{well-formed} if it satisfies the following:
     
    \begin{enumerate}[leftmargin=*]
        \item \textit{Procedures only}: 
        $\Prog$ consists of only procedures.
        \label{item:wf:procedures}
        

        \item \textit{Calling convention:}
        $\Prog$ defines a \textit{calling convention} $\ArgFn \!: \InstMemSet \!\pto\! \GprSet$, a map from $\mathtt{CALL}$/\allowbreak$\mathtt{RET}$ instructions to the set of registers used to pass arguments.
        \label{item:wf:cc}
    
        \item \textit{Data stack}:
        $\Prog$ has a \textit{data stack} $\Stack \subseteq \DataMemSet$, represented as a set of data addresses, such that
        (a) the stack pointer $\SP$ is always public and always points into $\Stack$ in all traces; 
        (b) no global variables are in stack memory;\footnote{I.e., no CA global loads/stores access the stack $\mathit{DS}$.} and 
        (c) $\Stack$ is zero-initialized in all initial configurations.\label{item:wf:stack}
        \item \textit{No callee-saved registers:}
        No general-purpose registers are preserved across procedure calls.
        \label{item:wf:live}
        
        \item \textit{No segfaults}:
        No sequential traces access an unmapped data address outside the stack $\Stack$.
        \label{item:wf:segfault}
    \end{enumerate}
\end{definition}

\subsubsection{Security-Typeable}
\label{sec:security-typeable}
Intuitively, a well-formed program $\Prog$ is security-typeable if each variable can be statically typed with a security label that is never violated at runtime.\footnote{\textit{Security-typeable} is implicitly defined with respect to $\Prog$'s initial configuration set $\InitConfSet$, which defines $\Prog$'s security policy (\S\ref{sec:model:program}).}
Note that \tool{} does not require such a security labeling to be explicitly provided as input; \textit{\tool{} does not require any program annotations whatsoever}.
Instead, \tool{} relies on the following properties afforded by security-typeable \ourct{} programs to infer an upper bound on what program variables may hold secret values (i.e., secretly-labeled values, \S\ref{sec:model:labels}) at runtime.


\begin{definition}[\textbf{\RuleTT{}}]
\label{def:security-typeable}
    A program $\Prog$ is \textit{security-typeable} if there exists 
        (i) a \textit{global variable typing} $\STGlob \!: \mathcal{M}_G \!\to\! \LabelSet$ where $\mathcal{M}_G \!\subseteq\! \DataMemSet$ is the set of data addresses accessed by CA global accesses (\Cref{def:constant-address});
        (ii) a \textit{stack variable typing} $\STStk{F} \!: [0, \FrameSize{F}) \!\to\! \LabelSet$ for each procedure $F$ (recall $\FrameSize{F}$ is $F$'s frame size); and
        (iii) a \textit{register variable typing} $\STReg{F} \!: \RegSet \times F \!\to\! \LabelSet$ for each procedure $F$.
    The typings must also pass the following typechecking rules (where $I$ is an instruction in some procedure $F$):
    \begin{enumerate}[leftmargin=*]
        \item \textit{Conservative:}\label{item:tt:conservative}
        No security type is violated in any sequential trace---i.e., no publicly-typed global, stack, or register variable ever holds a secret value, 
        and no secret values are read from or written to a publicly-typed stack variable.
        
        \item \textit{Transmitters:}\label{item:tt:xmit}
        All sensitive register operands of transmitters (\S\ref{sec:model:transmitters}) are publicly-typed in $\STReg{F}$.

        \item \textit{Load consistency:}\label{item:tt:load-consistency}
        The destination register of a CA load $I \mapsto \Load \, [\ZR/\SP + d], r$ from a secretly-typed global/stack variable is secretly-typed in $\STReg{F}$.

        \item \textit{Store consistency:}\label{item:tt:store-consistency}
        The source register of a CA store $I \mapsto \Store \, [\ZR/\SP + d], r$ to a publicly-typed global/stack variable is publicly-typed in $\STReg{F}$.

        \item \textit{Policy consistency:}\label{item:tt:policy-consistency}
        Publicly-typed global variables in $\STGlob$ contain public values in initial configurations of $\Prog$.

        \item \textit{Always-public registers:}\label{item:tt:always-public}
        The stack pointer $\SP$ and program counter $\PC$ are always publicly-typed in $\STReg{F}$.

        \item \textit{Register deps:}\label{item:tt:dep}
        The output of an $\Op$ instruction is secretly-typed in $\STReg{F}$ iff any input is secretly-typed in $\STReg{F}$.
    
        \item \textit{No-op:}\label{item:tt:nop} 
        Types of registers unmodified by instruction $I$ do not change in $\STReg{F}$ from $I$ to $J \in \Succs(I)$. 

        \item \textit{Public arguments:}\label{item:tt:arg}
        All argument registers $\ArgFn(I)$ at each $I \mapsto \MyCall \mid \Ret$ are publicly-typed in $\STReg{F}$.
    \end{enumerate}
\end{definition}

\RuleTT{}.\ref{item:tt:conservative}--\ref{item:tt:xmit} imply a \ctprog{} program in the traditional sense (per \Cref{def:constant-time}, proven in \S\ref{app:sec:cts->ct}).
\RuleTT{}.\ref{item:tt:conservative}--\ref{item:tt:nop} codify the static security properties of \ctprog{} programs that prior software-based mitigations implicitly assume of their input~\cite{vassena:blade} and CT compilers guarantee of their output~\cite{FactLanguage}, resolving the first limitation in \S\ref{sec:properties:ct-limitations}.
\RuleTT{}.\ref{item:tt:arg}
resolves the second limitation in \S\ref{sec:properties:ct-limitations}
by requiring that \textit{all secrets arguments be passed by reference} rather than by value; that is, one must pass a public pointer to a secret rather than the secret value itself.


\subsection{\OurPrimitives{} in a \ourct{} Program}

We now provide a complete characterization of Spectre leakage in \ourct{} programs, which gives rise to a powerful Spectre mitigation approach (\S\ref{sec:complete-spectre-mitigation}).

\subsubsection{Dynamic DFG}
\label{sec:dynamic-dfg}
The \textit{dynamic data-flow graph (DFG)} for a trace $e$ is a directed acyclic graph where nodes are register-step pairs $(r, i)$ and edges $(r,i) \!\DynDep{e}\! (r',j)$ encode direct \textit{dynamic} register or memory dependencies in the trace as follows:
\begin{itemize}[leftmargin=*]
    \item \textit{No-op:} $(r,i) \DynDep{e} (r,i+1)$ if $r$ is not modified by $I_i$.
    \item \textit{Register dependency:} $(r,i) \DynDep{e} (r',i+1)$ if $I_i \mapsto \Op \, r', \Vec{r_s}$ and $r \in \Vec{r_s}$.
    \item \textit{Memory dependency:}
    $(r,i) \DynDep{e} (r',j+1)$ if a store $I_i \mapsto \Store \, [r_a+d],r$ sources a later load $I_j \mapsto \Load \, [r_a'+d'],r'$.
    Unlike the prior register dependencies, memory dependencies can span many steps in the trace, since the store/load may not execute consecutively.
\end{itemize}
We say $(r',j)$ is \textit{dynamic-dependent} on $(r,i)$ if $(r,i) \!\DynDeps{e}\! (r',j)$, i.e., 
    $(r,i) \!\DynDep{e}\! \!\cdots \!\DynDep{e}\! (r',j)$.
    We say $(r', j)$ is \textit{dynamic-dependent on a load} $I_i \mapsto \Load \, [r_a\texttt{+}d],r$
    if $(r,\underline{i+1}) \DynDeps{e} (r',j)$.

We use a one-step delay $(r,i+1)$
to reference the output of an instruction $I_i$ updating register $r$ (e.g., a load $I_i \mapsto \Load \, [r_a+d],r$).
This is because the updated register does not hold its new value until the \textit{next} configuration, $C_{i+1}$.


\subsubsection{Taint Primitives}
\label{sec:properties:ourprimitives}
\label{sec:ourprimitives}
\lstdefinestyle{intro}{
    language=c,
    basicstyle=\ttfamily\footnotesize,
    escapechar=|,
    frame=single
}

\newcommand{\secexpr}[1]{\textcolor{black}{\hl{#1}}}
\newcommand{\leakexpr}[1]{\textcolor{red}{#1}}
\newcommand{\prim}[1]{\underline{#1}}

\begin{figure}[t]
    \centering
    \footnotesize

    \begin{minipage}{0.45\linewidth}
        \begin{lstlisting}[style=intro, title={\footnotesize \PrimNCAL{}: publicly-typed \\ non-constant-address load}]
|\secexpr{x}| = |\prim{*p}|;
temp = |\leakexpr{A[\secexpr{x}]}|;
        \end{lstlisting}
        \vspace{-4mm}
        \begin{lstlisting}[style=intro, title={\footnotesize \PrimNCAS{}: secretly-typed \\ non-constant-address store}]
x = 0;
*p = |\secexpr{secret}|;
temp = |\leakexpr{A[\textcolor{black}{\prim{\secexpr{x}}}]}|;
        \end{lstlisting}
    \end{minipage}
    \hfill
    \begin{minipage}{0.45\linewidth}
        \centering
        \begin{lstlisting}[style=intro, title={\footnotesize \PrimStack{}: publicly-typed \\ uninitialized stack load}]
x = 0;
temp = |\leakexpr{A[\textcolor{black}{\prim{\secexpr{x}}}]}|;
        \end{lstlisting}
        \vspace{-4mm}
        \begin{lstlisting}[style=intro, title={\footnotesize \PrimArg{}: unexpectedly \\ secret argument}]
|\secexpr{y}| = |\secexpr{secret}|;
|\prim{f()}|;
qux(int |\secexpr{x}|):
  temp = |\leakexpr{A[\secexpr{x}]}|;
        \end{lstlisting}
    \end{minipage}
    
    \caption{Exactly four kinds of \textit{\ourprimitives{}} introduce transient security type violations in \ourct{} programs, given our hardware model (\S\ref{sec:hardware-model}).
    \Ourprimitives{} can be enabled by \textit{any} speculation primitive.
    \prim{\Ourprimitives{}} are underlined; \secexpr{secrets} are highlighted; \leakexpr{transmitters} are red.
    }    
    \label{fig:primitives}
\end{figure}

\label{sec:security-type-violations}
\begin{definition}
\label{def:security-type-violation}
A \textit{security-type violation} is a register-step pair $(r, i)$ where $r$ is publicly-typed at $I_i$ (i.e., $\STReg{F_i}(r, I_i) = \Pub$, \Cref{def:security-typeable}) but $r$ holds a secret value at step $i$ (i.e., $R_i(r) = v_\Sec$, \S\ref{sec:model:labels}). 
\end{definition}

\begin{definition}[\Ourprimitives{}]
\label{def:ourprimitives}
    Instruction $I_i$ executing at step $i$ in trace $e$ is a \textit{\ourprimitive{}} if there is a security-type violation at step $i+1$ (\Cref{def:security-type-violation}) that is not dynamic-dependent (\S\ref{sec:dynamic-dfg}) on any prior security-type violations.
\end{definition}

Intuitively, \Cref{def:ourprimitives} says a \textit{\ourprimitive{}} is an instruction whose execution introduced a \textit{new} security-type violation
into the computation, since no inputs to $I_i$ violated their security types but the \textit{output} of $I_i$ did.
Now, we prove that \ourct{} programs in \ourmodel{} contain exactly \textit{four classes} of \ourprimitives{} (see \Cref{fig:primitives}).

\begin{theorem}[\Ourprimitives{}]
\label{thm:origin}
Every \ourprimitive{} $I_i$ in any trace $e$ can be classified as one of the following \textit{transient instructions} (with the register violating its security type at step $i+1$ in parentheses):
\begin{itemize}[leftmargin=*]
    \item \textbf{\PrimNCAL{}:} a transient \underline{NCA} \underline{l}oad (output register).\label{item:origin:ncal}
    \item \textbf{\PrimNCAS{}:} a transient CA load that reads from a transient \underline{NCA} \underline{s}tore (output register).\label{item:origin:ncas}
    \item \textbf{\PrimStack{}:}\label{item:origin:stack} a transient CA \underline{st}ac\underline{k} \underline{l}oad (output register). 
    \item \textbf{\PrimArg{}:} a transient $\MyCall{}/\Ret{}$ (\underline{n}on-\underline{arg}ument register).\label{item:origin:arg}
\end{itemize}
\end{theorem}

\begin{proof}
    We will prove the claim directly.
    Let $I_i$ be a \ourprimitive{} in a trace $e$.
    By \Cref{def:ourprimitives}, $I_i$ introduces a new security-type violation in some register $r'$ at step $i+1$.

    \textbf{Suppose $I_i$ does not modify $r'$.}
    Then $(r',i) \!\DynDep{e}\! (r',i+1)$ and $r'$ is secretly-typed at $i$ but publicly-typed at $i\!+\!1$ (\Cref{def:ourprimitives}).
    Since the security type of $r'$ cannot change intraprocedurally if $r'$ is unmodified (\RuleTT{}.\ref{item:tt:nop}),
    $I_i \mapsto \MyCall \mid \Ret$.
    Secretly-typed $r'$ cannot be an argument register (\RuleTT{}.\ref{item:tt:arg}),
    so $I_i$ satisfies \textbf{\PrimArg{}}.

    \textbf{Suppose $I_i$ modifies $r'$.}
    First, note that $r' \neq \PC$ since \ourmodel{}'s declassification of transmitter operands (\S\ref{sec:leakage-model}) ensures that $\PC$ always holds a public value.
    Only two instructions modify non-$\PC$ registers: $\Op$ and $\Load$.
    
    If $I_i \!\mapsto\! \Op_o \, r', \Vec{r_s}$, then some input $r \in \Vec{r_s}$ violated its security type at $i$ (\RuleTT{}.\ref{item:tt:dep}).
    But $(r,i) \DynDep{e} (r',i+1)$, 
    so $I_i$ is not a taint primitive (\Cref{def:ourprimitives}), a contradiction.

    If $I_i \mapsto \Load \, [r_a+d],r'$, then 
    $I_i$ may be an NCA, CA stack, or CA global load.
    If $I_i$ is an NCA load (resp. CA stack load), it satisfies \textbf{\PrimNCAL{}} (resp. \textbf{\PrimStack{}}).
    If $I_i$ is a CA global load, 
    we know it did not read from a CA stack store (\RuleWF{}.\ref{item:wf:stack}) or initial memory (\RuleTT{}.\ref{item:tt:conservative}), so it read from an NCA store or CA global store.
    In the former case, the store was not sequential (\RuleTT{}.\ref{item:tt:conservative}), satisfying \textbf{\PrimNCAS{}}.    
    In the latter, $I_i$ read from a CA global store of a prior security type violation (\RuleTT{}.\ref{item:tt:store-consistency}), a contradiction (\Cref{def:ourprimitives}).
\end{proof}

\subsubsection{Mitigating Spectre in \ourct{} Programs}
\label{sec:complete-spectre-mitigation}


\begin{corollary}
\label{cor:conditions-for-sct-violations}
If a \ourct{} program $\Prog$ violates SCT (\Cref{def:speculative-constant-time}), then there exists a trace in which a transmitter's sensitive operand is dynamic-dependent on an \PrimNCAL{}, \PrimNCAS{}, \PrimStack{}, or \PrimArg{} \ourprimitive{}. 
\end{corollary}
\begin{proof}[Proof]
If $\Prog$ violates SCT, some trace $e$ of $\Prog$ executes a transmitter $I_j$ at step $j$ with a secretly-labeled sensitive operand $r$ (\Cref{def:speculative-constant-time}).
By \RuleTT{}.\ref{item:tt:xmit}, $r$ is publicly-typed at $I_j$, so $(r,j)$ is a security-type violation (\Cref{def:security-type-violation}).
Thus, $I_j$ is dynamic-dependent on some \ourprimitive{} $I_i$ (\Cref{def:ourprimitives}).
By \Cref{thm:origin}, $I_i$ is a \PrimNCAL{}, \PrimNCAS{}, \PrimStack{}, or \PrimArg{}.
\end{proof}

\section{\tool{}: A Compiler Approach for Enforcing Speculative Constant Time}
\label{sec:tool}
\Cref{cor:conditions-for-sct-violations} is powerful: to eliminate \textit{all Spectre leakage} in a \ourct{} program, a mitigation must simply break all dynamic dependencies from four classes of taint primitives (\Cref{thm:origin}) to subsequent transmitters.
\tool{} consists of three intraprocedural passes (\Cref{fig:tool:overview}) that do just this.

\begin{figure*}
    \centering
    {\footnotesize\begin{tikzpicture}
        \def\PrimsDist{0.7cm}
        \node[draw, text width=1.3cm, align=center] (prog-in) {
            $\ProgIn$ \\
            \cmarkc{} [\ourct{}] \\
            \colorbox{yellow}{\xmarkc{} [SCT]}
        };
        \node[below=\PrimsDist of prog-in] (conds-in) {
            \begin{tabular}{ll}
                \xmarkc{} \PrimNCAL{} & \xmarkc{} \PrimStack{} \\
                \xmarkc{} \PrimNCAS{} & \xmarkc{} \PrimArg{}
            \end{tabular}
        };
        \draw[->, double equal sign distance] (prog-in.south) -- (conds-in.north) node[right, pos=0.5, xshift=0.1cm, text width=1.5cm] {Thm.~\ref{thm:origin}\\Cor.~\ref{cor:conditions-for-sct-violations}};

        \node[draw, text width=1.2cm, align=center, right=2.5cm of prog-in.east] (prog-fence) {
            $\ProgFence$ \\
            \cmarkc{} [\ourct{}] \\
            \xmarkc{} [SCT]
        };
        \node[below=\PrimsDist of prog-fence] (conds-fence) {
            \begin{tabular}{ll}
                \pmarkc{} \PrimNCAL{} & \xmarkc{} \PrimStack{} \\
                \cmarkc{} \PrimNCAS{} & \xmarkc{} \PrimArg{}
            \end{tabular}
        };
        \draw[->, double equal sign distance] (prog-fence.south) -- (conds-fence.north) node[left, pos=0.5, xshift=-0.1cm] {\Cref{thm:no-ncas}};

        \node[draw, text width=2.75cm, align=center, right=2.5cm of prog-fence.east] (prog-fps) {
            $\ProgFPS$ \\
            \begin{tabular}{l}
            \hspace{-0.25cm} \cmarkc{} [\ourct{}] (\Cref{cor:fps-cts}) \\
            \hspace{-0.25cm} \xmarkc{} [SCT]
            \end{tabular}
        };
        \node[below=\PrimsDist of prog-fps] (conds-fps) {
            \begin{tabular}{ll}
                \pmarkc{} \PrimNCAL{} & \cmarkc{} \PrimStack{} \\
                \cmarkc{} \PrimNCAS{} & \xmarkc{} \PrimArg{}
            \end{tabular}
        };
        \draw[->, double equal sign distance] (prog-fps.south) -- (conds-fps.north) node[left, pos=0.5, xshift=-0.1cm] {Thm.~\ref{thm:no-stkl}};

        \node[draw, text width=1.3cm, align=center, right=2.5cm of prog-fps.east] (prog-out) {
            $\ProgOut$   \\
            \cmarkc{} [\ourct{}] \\
            \colorbox{yellow}{\cmarkc{} [SCT]}
        };
        
        \node[below=\PrimsDist of prog-out] (conds-out) {
            \begin{tabular}{ll}
                \cmarkc{} \PrimNCAL{} & \cmarkc{} \PrimStack{} \\
                \cmarkc{} \PrimNCAS{} & \cmarkc{} \PrimArg{}
            \end{tabular}
        };
        \draw[->, double equal sign distance] (prog-out.south) -- (conds-out.north) node[left, pos=0.5, xshift=-0.1cm, text width=1cm] {Thms.\\\ref{thm:no-arg}, \ref{thm:tool-correct}};

        \draw[->] (prog-in.east) -- (prog-fence.west) node[above, text width=2cm, pos=0.5, align=center] {Fence Insertion\\(\S\ref{sec:tool:fence})};
        \draw[->] (prog-fence.east) -- (prog-fps.west) node[above, text width=2.3cm, pos=0.5, align=center] {Function-Private\\Stacks~(\S\ref{sec:tool:fps-pass})};
        \draw[->] (prog-fps.east) -- (prog-out.west) node[above, text width=2cm, pos=0.5, align=center] {Register Zeroing~(\S\ref{sec:tool:arg-pass})};

        \draw[color=gray, dashed]   ($(prog-in.east)+(0.25,0)$) to[out=-90,in=90] 
                ($(prog-in.south east)!0.4!(prog-fence.south west)$) to[out=-90,in=90]
                ($(conds-in.south east)!0.5!(conds-fence.south west)$) to[out=0,in=180,looseness=0.1]
                ($(conds-fps.south east)!0.5!(conds-out.south west)$) to[out=0,in=-90]
                ($(prog-out.west)-(0.25,0)$) to[out=90,in=0]
                ($(prog-out.north west)+(-0.5,0.2)$) to[out=180,in=0]
                ($(prog-in.north east)+(0.5,0.2)$) to[out=180,in=90]
                ($(prog-in.east)+(0.25,0)$);
        \node[draw=gray, dashed] at ($(prog-fps.south east)+(1,-0.5)$) {\normalsize \tool{}};

    \end{tikzpicture}}
    
    \caption{
    Given a \ourct{} program $\ProgIn$ (\S\ref{sec:propreties:ourct}), \graydashedbox{\tool{}}
    runs three passes, each of which provably eliminate a class of \ourprimitives{} (\S\ref{sec:properties:ourprimitives}).
    \Cref{thm:tool-correct} proves all passes together eliminate the final primitive, \PrimNCAL{} and thus the output \ourct{} program $\ProgOut$ satisfies SCT (\Cref{def:speculative-constant-time}).
    }
    \label{fig:tool:overview}
\end{figure*}

\subsection{\Tool{}'s Fence Insertion Pass}
\label{sec:tool:fence}
\tool{} runs the
\textit{Fence Insertion Pass} first to mitigate \textit{all} SCT violations due to \PrimNCAS{}
and \textit{some} due to \PrimNCAL{} \ourprimitives{} \textit{in each procedure $F$}.
We formulate optimal fence insertion as a graph cut problem over $F$'s weighted transient CFG (\S\ref{sec:tool:fence:tcfg}): we must eliminate all transient control-flow paths from \textit{sources} to \textit{sinks}
(defined in \S\ref{sec:tool:fence:source-sink-pairs}).

\subsubsection{Transient CFG}
\label{sec:tool:fence:tcfg}
First, \tool{} constructs a weighted \textit{transient CFG} (T-CFG) for $F$ which captures the set of all transient control-flow paths through $F$. 
Unlike a traditional procedural CFG, the T-CFG captures transient execution through $F$ \textit{and} across invocations of $F$ or spurious \rsb{}-mispredicted returns to $F$.
Nodes are instructions in $F$. Edges are defined by the \textit{transient successor} function:
\begin{align*}
\mathbf{I_\mathrm{enter}} &= \{J \!\in\! F \mid J \!\mapsto\! \Endcall \mbox{ or } J\!-\!1 \!\mapsto\! \MyCall\} \\
\mathbf{I_\mathrm{exit}} &= \{I \in F \mid I \mapsto \MyCall \mid \Ret\} \\
\TSuccs(I) &= \begin{cases}
    \mathbf{I_\mathrm{enter}} &\mbox{if } I \in \mathbf{I_\mathrm{exit}} \\
    \emptyset &\mbox{if } I \mapsto \Lfence \\
    \Succs(I) &\mbox{otherwise (\S\ref{sec:successors})}
\end{cases}
\end{align*}
The T-CFG has the edge $\smash{I \to_\mathrm{tcfg}^F J}$ iff $\smash{J \in \TSuccs(I)}$.
We write $I \to_\mathrm{tcfg*}^F J$ to indicate there is a path from $I$ to $J$ in $F$'s T-CFG.
Notably, $\Lfence$s have no transient successors since they block transient execution (\S\ref{sec:model:semantics:lfence}) and are thus dead-ends in the T-CFG. 
We compute the \textit{weight} of an edge $I \to_\mathrm{tdfg}^* J$ as $w = L/D$, where $L$ is the loop nest depth and $D$ is the depth in the dominator tree of instruction $J$.
This estimates the relative execution frequency of $J$ and thus the cost of placing a fence before it.
Weights affect optimality 
(performance)
but not correctness of the min-cut (\S\ref{sec:tool:fence:multicut}).

\begin{theorem}[T-CFG complete]
\label{lem:tcfg-complete}
If same-procedure instructions $I_i, I_j \in F$ transiently execute at steps $i < j$ of some trace of a program $\Prog$,
then $I_i \to_\mathrm{tcfg*}^F I_j$. (\textit{Proof.} See \S\ref{app:sec:tcfg-complete}.)
\end{theorem}

\subsubsection{Static DFG}
\label{sec:static-dfg}
Next, \tool{} constructs a \textit{static DFG} for $F$, which models syntactic intraprocedural dependencies through CA stack accesses and registers.
Nodes are register-instruction pairs $(r,I) \!\in\! \RegSet \!\times\! F$. 
Edges $(r,I) \!\StcDep{F}\! (r',J)$ encode direct register or stack dependencies as follows:
\begin{itemize}[leftmargin=*]
    \item \textit{No-op:} 
    $(r,I) \StcDep{F} (r,J)$ if $J \in \TSuccs(I)$ and $I$ does not modify $r$.
    
    \item \textit{Register dep:}
    $(r,I) \StcDep{F} (r',I+1)$ for each $r \in \Vec{r_s}$ if $I \mapsto \Op_o r', \Vec{r_s}$.

    \item \textit{Stack dep:}
    $(r,I) \StcDep{F} (r',J+1)$ if $I \mapsto \Store\,[\SP+d],r$ and $J \mapsto \Load\,[\SP+d],r'$;
    i.e., if $I$ and $J$ are a same-offset CA stack store and load, then $J$ may read from $I$.
\end{itemize}
We say $(r',J)$ is \textit{static-dependent} on $(r,I)$ if $(r,I) \!\StcDeps{F}\! (r',J)$. 
Furthermore, we say $(r',J)$ is \textit{static-dependent on a load} $I \mapsto \Load \, [r_a+d],r$
if $(r, I+1) \StcDeps{F} (r',J)$.

The static DFG enables \tool{} to precisely track how candidate \PrimNCAL{} \ourprimitives{}
may propagate security-type violations through intraprocedural dependencies.
Note the similarities with dynamic-dependencies (\S\ref{sec:dynamic-dfg}),
like the one-step delay for static DFG nodes referencing an instruction's \textit{output}.
We use these similarities in our final proof of \tool{}'s correctness (\Cref{thm:tool-correct}).

\subsubsection{Source-Sink Pair Generation}
\label{sec:tool:fence:source-sink-pairs}
\tool{} identifies \textit{five types} of intraprocedural \textit{source-sink} instruction pairs that can produce SCT violations involving an \PrimNCAL{} or \PrimNCAS{} primitive \textit{when both source and sink execute transiently}.
The source of each pair is an NCA load/store, which we conservatively assume may read/write a secret at an \textit{arbitrary} data address when executed transiently. \tool{} accumulates a set $S$ of source-sink pairs as follows.
A source-sink pair $(I,J)$ is added to $S$ for each static dependency $(r,I) \StcDeps{F} (r',J)$, where $I \mapsto \Load\,[r_a+d],r$ ($I$ is an NCA load) and $J$ is (i) a transmitter with sensitive operand $r'$ (\ncalxmit{});
(ii) a $\MyCall/\Ret$ with argument $r' \in \ArgFn(J)$ (\ncalarg{}); or
(iii) a CA global store of $r'$ (\ncalglob{}).
These pairs help \tool{} prevent \PrimNCAL{} \ourprimitives{} from passing secrets (i) intraprocedurally to transmitters or (ii) interprocedurally in arguments, or (iii) enabling stores of secrets to publicly-typed global variables.

A source-sink pair $(I,J)$ is added to $S$ for each NCA store $I$ paired with (iv) each CA load $J$ (\ncascal{}) and (v) each $\MyCall/\Ret$ $J$ (\ncasctrl{}).
These pairs help \tool{} prevent all candidate \PrimNCAS{} instructions (CA loads) from (iv) intraprocedurally or (v) interprocedurally reading from a transient NCA store $I$.

%

\subsubsection{Fence Insertion}
\label{sec:tool:fence:multicut}
\label{sec:tool:fence:insertion}
Given the set of source-sink pairs $S$ (\S\ref{sec:tool:fence:source-sink-pairs}) and the T-CFG for $F$ (\S\ref{sec:tool:fence:tcfg}), \tool{} runs a heuristic minimum directed multicut algorithm (\S\ref{app:sec:multicut}) to obtain a near-optimal edge cutset $\Cutset$ that prevents (transient) sources from passing data to sinks.
That is, \tool{} inserts an $\Lfence$ along each edge $(u,v) \in \Cutset$.

\subsubsection{Guarantees}
\tool{}'s Fence Insertion Pass produces a partially mitigated \ourct{} program $\ProgFence$ that 
contains \textit{no} SCT violations caused by \PrimNCAS{} \ourprimitives{} (\cmarkc{} in \Cref{fig:tool:overview}).
Some \PrimNCAL{} primitives remain (\pmarkc{}), which subsequent passes will eliminate.

\begin{theorem}[Sources execute sequentially]
\label{thm:source-sink}
    If $(I_i, I_k)$ is a source-sink pair in a trace of $\ProgFence$, $I_i$ did not trap, and $i < k$, then some $I_j \mapsto \Lfence$ executes between $I_i$ and $I_k$, so $I_i$ executed sequentially. (\textit{Proof.} See \S\ref{app:sec:source-sink}.)
\end{theorem}

\begin{theorem}[No \PrimNCAS{} in $\ProgFence$]\label{thm:no-ncas}
    No instructions belong to \PrimNCAS{} (\S\ref{sec:properties:ourprimitives}) in any trace of $\ProgFence$.
\end{theorem}
\begin{proof}
Suppose for contradiction some $I_k$ satisfies \PrimNCAS{} in trace $e$ of $\ProgFence$, i.e., $I_k$ is a CA load that read from a prior transient NCA store $I_i$.
If an instruction $I_j \mapsto \MyCall/\Ret$ executes for $i \!<\! j \!<\! k$, then $(I_i,I_j)$ is a \ncasctrl{} source-sink pair.
Else, $(I_i, I_k)$ is a \ncascal{} pair.
Either way, $I_i$ executes sequentially (\Cref{thm:source-sink}), a contradiction.
\end{proof}

\subsection{\Tool{}'s Function-Private Stacks Pass}
\label{sec:tool:fps-pass}
\tool{} runs the \textit{Function-Private Stacks (FPS) Pass} second to eliminate all \PrimStack{} \ourprimitives{} in $\ProgFence$, producing a functionally-equivalent program $\ProgFPS$.
Our insight is that SCT violations due to \PrimStack{} arise due to procedures reallocating frames on the same data stack. Thus, \tool{} statically assigns a \textit{private stack} to each procedure $F$ on which only $F$ can allocate stack frames.

\subsubsection{Motivation}
If procedures share the same data stack, it is difficult to ensure that a publicly-typed stack access will not read a secret value transiently.
\tool{}'s solution is to \textit{assign each procedure its own data stack}, which is never reused by another procedure.
Thus, a publicly-typed CA stack load will \textit{never} transiently read a value from a different-offset or different-procedure CA stack store.

\subsubsection{Implementation}
\label{sec:tool:fps-implementation}
\begin{figure}[tbp]
\begin{lstlisting}[language=c, numbers=left, numbersep=3pt, xleftmargin=15pt, basicstyle=\ttfamily\footnotesize, escapechar=|]
foo: ENDBR|\label{line:fps:prolog:endcall}|
   + LD [ZR+|$\psp{F}$|],SP // load private SP |\label{line:fps:prolog:ld}\label{line:fps:prolog:first}|
     SUB SP,SP,|$\FrameSize{F}$|   // frame allocation |\label{line:fps:prolog:sub}|
   + MAX SP,SP,|$B_F$|   // prevent overflow |\label{line:fps:prolog:probe}|
   + ST [ZR+|$\psp{F}$|],SP // store private SP |\label{line:fps:prolog:st}\label{line:fps:prolog:last}|
     ...
     CALL r1
   + LD [ZR+|$\psp{F}$|],SP // load private SP |\label{line:fps:call:ld}|
     ...
     ADD SP,SP,|$\FrameSize{F}$|    // frame deallocation |\label{line:fps:epilog:dealloc}|
   + MIN SP,SP,|$E_F$|   // prevent underflow |\label{line:fps:epilog:probe}\label{line:fps:epilog:first}|
   + ST [ZR+|$\psp{F}$|],SP // store private SP |\label{line:fps:epilog:st}\label{line:fps:epilog:last}|
     RET
\end{lstlisting}
\caption{
Instructions inserted by \tool{}'s FPS Pass (indicated with ``\texttt{+}'').
Lines \ref{line:fps:prolog:first}--\ref{line:fps:prolog:last} are the prologue;
lines \ref{line:fps:epilog:first}--\ref{line:fps:epilog:last} are the epilogue.
$\mathtt{MAX}/\mathtt{MIN}$ are instances of ASP's $\mathtt{OP}$ instruction (\S\ref{sec:model:other}).
}
\label{fig:tool:fps-transformation}
\label{fig:fps-example}
\end{figure}

Let $F$ be a procedure with stack frame size $\FrameSize{F}$.
\Cref{fig:tool:fps-transformation} depicts the code transformation performed by the FPS Pass.
%
Formally, given the input program $\ProgFence$ with data stack $\Stack_\mathrm{in}$, we define the output program $\ProgFPS = (\InstMemSet^\mathrm{fps}, \DataMemSet^\mathrm{fps}, P_\mathrm{fps}, \InitConfSet^\mathrm{fps})$ with data stack $\Stack_\mathrm{fps}$.
To start, we initialize $(\ProgFPS,\Stack_\mathrm{fps}) \gets (\ProgFence, \Stack_\mathrm{in})$.

\paragraph{\bf Private stack assignment}
\tool{} allocates and assigns a \textit{private stack} to $F$, denoted as $\fps{F}$.
To construct $\fps{F}$, we choose a stack base and stack end $\StackBase{F}, \StackEnd{F} \in \ValSet$ (where $\StackEnd{F} - \StackBase{F}$ is a positive multiple of $k_F$) and set
$$
\fps{F} \gets \underbrace{[\StackBase{F}, \StackEnd{F})}_{\mbox{\small \hspace{-3mm} usable region \hspace{-3mm}}} \cup \underbrace{[\StackEnd{F}, \StackEnd{F} + \FrameSize{F})}_{\mbox{\small underflow region}}
$$
such that $\fps{F} \cap \DataMemSet^\mathrm{fps} = \emptyset$.
We then map $\fps{F}$'s usable region into data memory ($\smash{  \DataMemSet^\mathrm{fps} \gets \DataMemSet^\mathrm{fps} \cup [\StackBase{F}, \StackEnd{F})  }$) and zero-initialize it in all initial configurations $C_0 \in \InitConfSet^\mathrm{fps}$;
however, we leave the underflow region unmapped, since well-formed procedures never sequentially underflow their stack.
We also add the private stack $\fps{F}$ to the program's stack metadata ($\Stack_\mathrm{fps} \!\gets\! \Stack_\mathrm{fps} \cup \fps{F}$, \Cref{def:well-formed}).
Finally, \tool{} allocates a global variable $\psp{F} \in \DataMemSet$ to hold the private stack pointer and initializes it to the stack end $D_0[\psp{F} \!\gets\! \StackEnd{F}]$ in all initial configurations $C_0 \!\in\! \InitConfSet^\mathrm{fps}$. 

\paragraph{\bf Switching private stacks}
\tool{} inserts instructions into the prologue and epilogue of $F$ to ensure it uses its private stack $\fps{F}$ rather than the shared stack. \textit{Prologue}: $F$ loads its private stack pointer $\psp{F}$ (L\ref{line:fps:prolog:ld} in \Cref{fig:tool:fps-transformation}) before frame allocation (L\ref{line:fps:prolog:sub}); a lower bounds clip  (L\ref{line:fps:prolog:probe}) prevents stack overflow; and $F$'s updated stack pointer is saved (L\ref{line:fps:prolog:st}). 
\textit{Post-call}: after each $\MyCall$, we switch from the callee's stack pointer back to $F$'s stack pointer (L\ref{line:fps:call:ld}). 
\textit{Epilogue}: 
after frame deallocation (L\ref{line:fps:epilog:dealloc}), we insert an upper bounds clip (L\ref{line:fps:epilog:probe}) to restrict any subsequent transient stack underflows to the underflow region; and $F$ restores its private stack pointer at procedure entry (L\ref{line:fps:epilog:st}).
Since $\StackEnd{F} \!-\! \StackBase{F}$ is a multiple of $\FrameSize{F}$, the bounds clips produce mutually $\FrameSize{F}$-aligned stack pointers, preventing STKL taint primitives.

\subsubsection{Guarantees}
\tool{}'s FPS Pass produces a \ourct{} program $\ProgFPS$ (\Cref{cor:fps-cts}) that has no SCT violations due to \PrimNCAS{}  (\Cref{thm:no-ncas}) or \PrimStack{} \ourprimitives{} (\Cref{thm:fps-no-stkl}).


\begin{theorem}[Private and aligned stacks]
\label{thm:fps-inbounds}
    At each CA stack access $I_i$ in any trace of $\ProgFPS$, the stack pointer points inside the current procedure $F$'s private stack $\fps{F}$ and is aligned to $F$'s frame size $\FrameSize{F}$.
    Formally, $R_i(\SP) \in \fps{F}$ and $R_i(\SP) = \StackEnd{F} - m \cdot \FrameSize{F}$ for some $m \geq 0$. (\textit{Proof.} See \S\ref{app:sec:fps-inbounds}.)
\end{theorem}

\begin{corollary}
\label{cor:fps-cts}
    $\ProgFPS$ satisfies \ourct{}.
\end{corollary}
\begin{proof}
    \Cref{thm:fps-inbounds} implies the stack pointer $\SP$ is always in-bounds of $\Stack_\mathrm{fps}$ and so $\ProgFPS$ satisfies \RuleWF{}.\ref{item:wf:stack}.
    All other \ourct{} properties trivially continue to hold for $\ProgFPS$.
\end{proof}

\begin{theorem}[No \PrimStack{} in $\ProgFPS$]\label{thm:fps-no-stkl}\label{thm:no-stkl}
    No \ourprimitives{} (\S\ref{sec:properties:ourprimitives}) belong to \PrimStack{} in any trace of $\ProgFPS$.
\end{theorem}
\begin{proof}[Proof]
    Suppose for contradiction some \ourprimitive{} $I_j \mapsto \Load\,[\SP \texttt{+} d],r$ satisfies \PrimStack{} in a trace of $\ProgFPS$.
    $I_j$ read from a prior secretly-typed same-address store $I_i \!\mapsto\! \Store\,[r_a'\!+\!d'],r'$,
    which is neither a transient NCA store (\Cref{thm:no-ncas}), sequential NCA store (violates \RuleTT{}.\ref{item:tt:store-consistency}), nor CA global store (\RuleWF{}.\ref{item:wf:stack}).

    Thus, $I_i$ is a CA stack store. 
    Let (eq. 1) $A = R_i(\SP)+d' = R_j(\SP)+d$ be the effective address of $I_i$ and $I_j$.
    $I_i$ and $I_j$ must be in the same procedure $F$ since they are both in-bounds accesses to the same private stack (\Cref{thm:fps-inbounds}).
    Also by \Cref{thm:fps-inbounds}, (eq. 2) $R_i(\SP)=\StackEnd{F} - m' \cdot \FrameSize{F}$ and (eq. 3) $R_j(\SP)=\StackEnd{F} - m \cdot \FrameSize{F}$ for some $m,m' \geq 0$.
    Eqs. 1--3 imply 
    $d' - d = \FrameSize{F} \cdot (m' - m)$
    and thus $d' - d = 0 \pmod{\FrameSize{F}}$.
    Since frame offsets are less than the frame size (\RuleWF{}.\ref{item:wf:stack}), $d = d'$, so $I_i$ and $I_j$ access the same stack variable $d$.
    Thus $d$ (\RuleTT{}.\ref{item:tt:load-consistency}) and $I_i$ (\RuleTT{}.\ref{item:tt:store-consistency}) are publicly-typed like $I_j$.
    We conclude $(r',i)$ is a security-type violation but $(r',i) \DynDep{e} (r,j+1)$, thus $I_j$ is not a \ourprimitive{} by \Cref{def:ourprimitives}, a contradiction.
\end{proof}

\subsection{\tool{}'s Register Cleaning Pass}
\label{sec:tool:arg-pass}
\tool{}'s final pass, the Register Cleaning Pass, eliminates all \PrimArg{} \ourprimitives{} from $\ProgFPS$ to produce the fully mitigated output program $\ProgOut$. 
It inserts instructions to zero out all non-argument registers (defined by the calling convention $\ArgFn$, \Cref{def:well-formed}) before each call and return.
Formally, let $I$ be a $\MyCall/\allowbreak\Ret$.
If $I \!\mapsto\! \MyCall \, r$, set $\mathbf{r_\textbf{zero}} = \GprSet \setminus\allowbreak \ArgFn(I) \setminus r$;
if $I \!\mapsto\! \Ret$, set $\mathbf{r_\textbf{zero}} = \GprSet \setminus \ArgFn(I)$.
For each $r_\text{zero} \in \mathbf{r_\textbf{zero}}$, insert $\mathtt{MOV}\,r_\text{zero},0$ directly before $I$.

% ...

\begin{theorem}[No \PrimArg{} in $\ProgOut$]\label{thm:no-arg}
    No trace of $\ProgOut$ features any instruction satisfying \PrimArg{}.
\end{theorem}
\begin{proof}
    We publicly zero all non-argument registers before each $\MyCall/\Ret$, so \PrimArg{} is never satisfied.
\end{proof}

\subsection{Proof of \tool{}'s Correctness}
\label{sec:tool:proof}
In \S\ref{sec:tool:fence}--\ref{sec:tool:arg-pass}, we proved that \tool{}'s three passes eliminate all \PrimNCAS{}, \PrimStack{}, and \PrimArg{} \ourprimitives{}.
Now, we prove in \Cref{thm:tool-correct} that the fully mitigated program $\ProgOut$ has no \PrimNCAL{} \ourprimitives{} and is thus 
SCT (\Cref{def:speculative-constant-time}), i.e., does not transiently leak secrets.

\begin{theorem}
\label{thm:tool-correct}
    $\ProgOut$ satisfies speculative constant-time. 
\end{theorem}



\begin{proof}[Proof]
    Suppose for contradiction $\ProgOut$ is not SCT.
    By \Cref{cor:conditions-for-sct-violations}, a transmitter $I_l$ with sensitive operand $\rxmit$ is dynamic-dependent (\S\ref{sec:dynamic-dfg}) on an \PrimNCAL{} \ourprimitive{} (i.e., a transient NCA load, \S\ref{sec:properties:ourprimitives}) in some trace $e$ of $\ProgOut$.
    Let $I_i \mapsto \Load\,[r_a+d],\rncal$ be the 
    \textit{most recent} transient NCA load to execute on which $(\rxmit, l)$ is dynamic-dependent,
    and let $F$ be the procedure containing $I_i$.
    That is, $(\rncal,i+1) \DynDeps{e} (\rxmit,l)$, and
    for all steps $k$ with $i < k < l$, $I_k \mapsto \Load\,[r_a'+d'],r' \implies (r',k+1) \not \DynDeps{e} (\rxmit,l)$.
     \textcolor{white}{be}
%
    Recall that we use a one-step delay $(\rncal, i+1)$ to describe dependencies to/from the output of a load $I_i$ (\S\ref{sec:dynamic-dfg}).
    
    We will show that there exists a source-sink pair $(I_i, I_j)$ ($i \!<\! j \!\leq\! l$) which guarantees that $I_i$ executed \textit{sequentially} via \Cref{thm:source-sink}, yielding a contradiction (since $I_i$ is \textit{transient} by assumption).
    There are \textit{two cases} we consider: 
    either $(\rxmit,I_l)$ \textit{is} or \textit{is not} static-dependent on $I_i$.

    %

    If $(\rxmit, I_l)$ \textit{is} static-dependent on $I_i$ (\S\ref{sec:static-dfg})---i.e., $(\rncal,I_{i+1}) \StcDeps{F} (\rxmit, I_l)$---then $(I_i, I_l)$ is a \ncalxmit{} source-sink pair and we are done.
    
    \textit{Otherwise}, there is some dynamic dependency from
    $I_i$ to $I_l$ not mirrored in a static dependency:
    \begin{alignat}{3}
        (\rncal, i\!+\!1)   &\DynDeps{e}    &(r&, j)     &&\!\DynDep{e}\! (r',k\!+\!1) \!\DynDeps{e}\! (\rxmit, l) \nonumber \\
        (\rncal, I_{i+1}) &\StcDeps{F}   &(r&,I_j)    &&\!\not\StcDep{F}\! (r', I_{k+1}) \label{eq:stc}
    \end{alignat}
    We will show that $(I_i, I_j)$ forms a source-sink pair. 
    Clearly, $(r,j) \DynDep{e} (r',k\!+\!1)$ is not an intraprocedural dynamic register dependency (\S\ref{sec:dynamic-dfg}), or else $(r,I_j) \StcDep{F} (r',I_{k+1})$ would hold. 
    Thus, it must be a dynamic interprocedural register dependency or memory dependency,
    implying (a) $I_j \mapsto \MyCall/\Ret$ or (b)--(d) $I_j$ is a store.

    \noindent \quad (a) Suppose $I_j \!\mapsto\! \MyCall/\Ret$. 
    The Register Cleaning Pass (\S\ref{sec:tool:arg-pass}) breaks all dependencies through \textit{non}-arguments,
    so $r$ must be an argument.
    Thus, $(I_i, I_j)$ forms a \ncalarg{} source-sink pair and we are done.

    \noindent \quad (b) If $I_j$ is an NCA store, $I_k$ is a load.
    $I_k$ must be a CA load, since we already assumed $I_i$ was the most recent NCA load.
    By \Cref{thm:no-ncas}, $I_j$ executed sequentially, thus so did $I_i$, contradicting that $I_i$ is a transient NCA load.

    \noindent \quad (c) If $I_j$ is a CA global store, then $(I_i, I_j)$ is a \ncalglob{} source-sink pair and we are done.
    
    \noindent \quad (d) Else, $I_j \!\mapsto\! \Store \, [\SP \texttt{+} d],r$ is a CA stack store.
    $I_k$ is not an NCA load (by assumption, $I_i$ is the most recent NCA load on which $I_l$ depends) or CA global (\RuleWF{}.\ref{item:wf:stack}).
    Thus, $I_k \!\mapsto\! \Load \, [\SP \texttt{+} d],r'$ is a same-offset CA stack load in $F$ (\Cref{thm:fps-inbounds}),
    so $(r,I_j) \!\StcDep{F}\! (r',I_k+1) = (r',I_{k+1})$ by the definition of static stack dependencies (\S\ref{sec:dynamic-dfg}), contradicting \Cref{eq:stc}.
\end{proof}

\section{Hardening Software Against Spectre} 
\label{sec:case-study}



%
\label{sec:mitigation-strategies}
We produce a code artifact \toolclang{}, an implementation of \tool{} for LLVM 14.
We empirically evaluate \toolclang{} on a suite of cryptographic primitives to assess its performance overhead.
For comparison, we evaluate \textit{three variants} of \toolclang{} alongside \textit{two baseline mitigations} and an \textit{insecure baseline} \textsc{none}.
See \S\ref{app:implementation-details} for additional LLVM-specific \toolclang{} implementation details.


\paragraph{\bf Baseline Mitigations}
\label{sec:case-study:baseline-mitigations}
Our baseline mitigations, \lfenceall{} and \slhall{}, layer Spectre mitigations discussed in \S\ref{sec:spectre-mitigations}.
\Cref{tab:mitigation-comparison} compares their security guarantees to \tool{}.
The \toolclang{} repository 
contains evaluations of two additional SLH-based mitigations based on BladeSLH~\cite{vassena:blade} and UltimateSLH~\cite{uslh}.

\paragraph{\bf \Tool{} Variants}
\label{sec:case-study:tool-variants}
To justify our hardware model (\S\ref{sec:hardware-model}), we implement three \toolclang{} variants: \toolclang{} (our proposal), \toolpsf{}, and \toolssbd{}.

\toolpsf{} implements a \tool{} extension that mitigates Spectre-\psf{} in software (\S\ref{app:model:psf}).
In general, \psf{} implicates \textit{all transient loads} as \ourprimitives{} (\S\ref{sec:properties:ourprimitives}), i.e., all loads may transiently return secrets, as long as the program has stored a secret since the last $\Lfence$.
Thus, we expand \tool{}'s \ncalxmit{} and \ncalarg{} source-sink pairs (\S\ref{sec:tool:fence:source-sink-pairs}) to encompass \textit{all} loads (not just NCA loads) as sources and omit the three other pair types.
The FPS Pass (\S\ref{sec:tool:fps-pass}) is not needed, since with \psf{} all stack loads may transiently return secrets.

\toolssbd{} implements a \tool{} extension where \stl{} is disabled, e.g., via \ssbd{}~\cite{intel:ssb}.
\toolssbd{} omits the FPS Pass, since its core motivation is mitigating Spectre-\stl{}.
It is replaced with a lighter-weight \textit{Stack Initialization Pass}, which zero-initializes the newly-allocated stack frame in a procedure's prologue.
A new kind of source-sink pair \textsc{call-xmit} is added to the Fence Insertion Pass (\S\ref{sec:tool:fence}), where sources are $\MyCall$s and sinks are transmitters that are dependent on CA stack loads. This pair captures that \rsb{} can cause a procedure to execute with the wrong stack frame featuring mismatching security types.


\paragraph{\bf Transmitters}
\toolclang{} assumes five transmitters: conditional branches, indirect branches (except returns), loads, stores, and division.
\ourmodel{} already models the first four (\S\ref{sec:model:transmitters});
\S\ref{sec:extension} shows how to extend \ourmodel{} to capture $\mathtt{DIV}$.


\begin{figure*}[t]
    \centering
    \includegraphics[width=\linewidth]{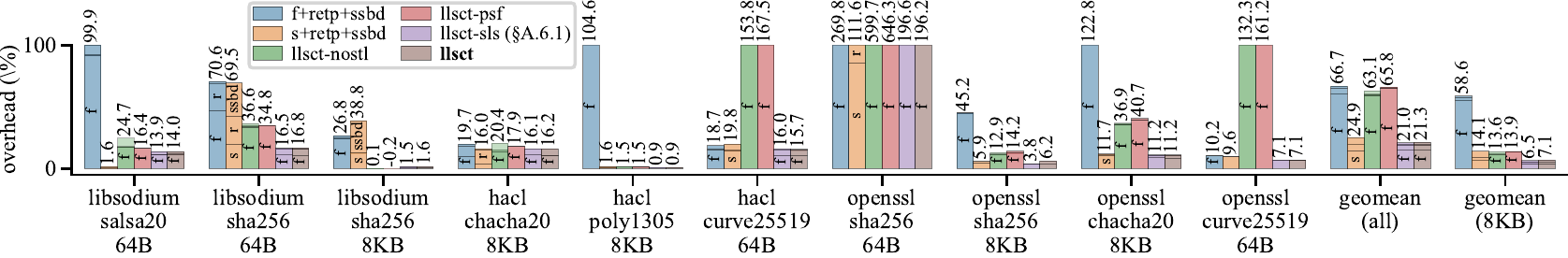}
    \caption{\small Runtime overheads of mitigations for crypto primitives in \libsodium{}, \hacl{}, and OpenSSL relative to code compiled with no mitigations. 
    Segments within the same bar indicate the additional overhead incurred by layering on the mitigation component.
    We label components with $\geq\!15\%$ overhead
    (``f'' = \lfence{}, ``s'' = \slh{}, ``r'' = \retpoline{}, ``ssbd'' = \ssbd{} speculation control).
    Total percent overhead is at the top of each bar. Overheads $> \!100\%$ are cut off.
     }
    \label{fig:benchmark-results}
\end{figure*}

\paragraph{\bf Compilation Setup}
All mitigations compile code with LLVM 
with \texttt{-O3} optimizations.
All \toolclang{} variants pass the \texttt{-mllvm -sct} flag to LLVM to enable relevant \tool{} mitigation passes.
We link all code using \textsc{LLD} \cite{llvm-lld}, which supports the \texttt{-z retpolineplt} flag (required for baseline mitigations using \retpoline{}). 
When compiling with \toolclang{} variants, we disable the following compiler optimizations which violate \ourct{} requirements (\S\ref{sec:propreties:ourct}):
\texttt{-mllvm -no\allowbreak{}-stack\allowbreak{}-slot\allowbreak{}-sharing},
\texttt{-mno\allowbreak{}-red\allowbreak{}-zone},
\texttt{-mllvm -no\allowbreak{}-argument\allowbreak{}-promotion}, \texttt{-fno\allowbreak{}-jump\allowbreak{}-tables}.
See \S\ref{app:unsafe-compiler} for details.

\paragraph{\bf Hardware Modes}
\label{sec:hardware-modes}
Each mitigation assumes a specific hardware model to uphold its assumptions (\S\ref{sec:hardware-model}), which we realize with a set of runtime flags  called the \textit{hardware mode}. We set the hardware mode before executing the program.
%
\textsc{none} assumes mode $\hwnone = \emptyset$, i.e., it does not restrict behavior.
%
\lfenceall{} and \slhall{}
assume mode $\hwssbd = \{\texttt{ssbd}\}$ which disables \stl{}.
\toolclang{} assumes mode $\hwasp = \{\texttt{doitm},\allowbreak \texttt{rrsbad},\allowbreak \texttt{psfd},\allowbreak \texttt{ibt},\allowbreak \texttt{shstk}\}$ to ensure the processor refines \ourmodel{} (\S\ref{sec:model}).
\toolssbd{} assumes mode $\hwaspssbd = \hwasp{} \cup \hwssbd{}$, which disables \stl{} in \ourmodel{}. 
\toolpsf{} assumes mode $\hwasppsf = \hwasp \setminus \{\texttt{psfd}\}$, which enables \psf{} in \ourmodel{}. 

\paragraph{\bf System and Workloads}
We run all experiments on a 24-core Alder Lake\footnote{%
While Alder Lake implements weak CET-IBT, 
we expect comparable performance on Alder Lake N, which implements strong CET-IBT (\S\ref{app:cet}).}
Intel\textregistered{} Core\texttrademark{} i9-12900KS processor with 640 KiB L1d, 768 KiB L1i, 14 MiB L2, and 30 MiB L3 caches and 128 GB RAM,
running a fork of Linux v5.9.0 that adds usermode IBT and SHSTK support~\cite{intel-cet-user}.
%

We evaluate crypto primitives from \libsodium{} (Salsa20, SHA-256)~\cite{libsodium}; the verified CT crypto library \hacl{} (ChaCha20, Poly1305, ECDH Curve25519)~\cite{hacl}; and OpenSSL (SHA-256, ChaCha20, ECDH Curve25519)~\cite{openssl}.


\section{Results}
\label{sec:results}


\Cref{fig:benchmark-results} compares the performance of our cryptographic benchmarks when mitigated with \toolclang{}, its two variants, and our two baseline mitigations. Performance is normalized to  \textsc{none}.
On average, \toolclang{} incurs \textit{lower overhead} (21.3\%) than both baseline mitigations 
(66.7\%/24.9\%) for \textsc{f}/\textsc{s}+\textsc{retp}+\ssbd{}).
\Cref{fig:benchmark-results} also decomposes \toolclang{}'s overhead into four components: Fence Insertion (``f,'' \S\ref{sec:tool:fence}), Function-Private Stacks (\S\ref{sec:tool:fps-pass}), Register Cleaning (\S\ref{sec:tool:arg-pass}), and \tool{}'s \hwasp{} hardware mode (\S\ref{sec:hardware-modes}).
Unsurprisingly, the overhead of fence insertion dominates overhead in all cases, whereas the overhead of FPS, register cleaning, and the \hwasp{} hardware model are negligible (and are thus unlabeled).


\paragraph{\bf Large Buffers}
Cryptographic code features arithmetic-heavy loops for common operations like encryption, decryption, hashing, and message authentication.
Thus, loop performance is crucial for overall cryptographic code performance.
\toolclang{} introduces \textit{half the overhead} (7.1\%) of the next best mitigation, \slhall{ } (14.1\%), for large buffer sizes (8\,KB).
This is because \tool{}'s loop-aware transient CFG construction (\S\ref{sec:tool:fence:tcfg}) 
prioritizes placing $\Lfence$s \textit{outside} of loops in its Fence Insertion Pass; both baselines \textit{require} placing mitigations \textit{inside} of loops.

%
\paragraph{\bf \Toolclang{} Variants}
\label{sec:results:variants}
Both \toolssbd{} and \toolpsf{} perform \textit{much worse} than \toolclang{}, in the worst case (OpenSSL SHA-256 64B) incurring 599.7\% and
646.3\% overhead, respectively.
Since neither variant uses FPS, they add new source-sink pairs (\S\ref{sec:case-study:tool-variants}) which require inserting more fences, \textit{tripling}  (63.1\%/65.8\%) the overhead relative to \toolclang{} (21.3\%).
Clearly, FPS and \tool{}'s default selection of source-sink pairs are essential to \toolclang{}'s efficiency.
Secondly, there is no clear way to adapt \tool{} to mitigate Spectre-\psf{} fully in software while maintaining \tool{}'s low overhead.
We conclude that the \ourmodel{} hardware model is the \textit{best fit} for the \tool{} approach, striking the ideal balance between 
enabling speculation primitives that can be efficiently mitigated in software (\pht{}, \btb{}, \rsb{}, \stl{}) and 
disabling those that cannot (\psf{}).

\paragraph{\bf Baseline Mitigations}
\label{sec:results:baseline}
As expected, \lfenceall{} performs worse than \slhall{} and \toolclang{} overall, since inserting $\Lfence$s after every conditional branch is expensive~\cite{lfenceoverhead}.
However, \slhall{} \textit{sometimes performs worse} than \lfenceall{}.
For \libsodium{}'s SHA-256 (8KB), \slhall{}'s overhead is over 10\%/35\% more than \lfenceall{}'s/\toolclang{}'s.
\Cref{fig:benchmark-results} shows that in this case, enabling \ssbd{} on top of \slh{}+\retpoline{} (which exhibits 12.7\% overhead) incurs 
    \textit{significant additional overhead} (26.1\%).
In contrast, enabling \ssbd{} for the insecure baseline \textsc{none} incurs $<1\%$ overhead.
This difference is likely due to the complexity that \slh{}'s masking operations introduce into the address calculation of stores; with \ssbd{}, stores must wait for longer to perform.

\paragraph{\bf Hardware Modes}
The top segment of each bar in \Cref{fig:benchmark-results} indicates the overhead of the mitigation's hardware mode (\S\ref{sec:hardware-modes}) when layered on top of its software mitigations.
The average overheads across all/8KB benchmarks for each are the following:
1.7\%/0.9\% for \lfenceall{} with \hwssbd{};
5.5\%/4.9\% for \slhall{} with \hwssbd{};
2.9\%/1.6\% for \toolssbd{} with \hwaspssbd{};
0.8\%/0.3\% for \toolpsf{} with \hwasppsf{}; and
1.9\%/2.3\% for \toolclang{} with \hwasp{}.

\section{Related Work and Conclusions}
\label{sec:related}
%
Several works study detection, formal foundations, and mitigation of Spectre attacks~\cite{survey:transient-exec-attacks,cauligi:sok-spectre-sw-defense}.

\paragraph{\bf Software Detection}
Symbolic execution~\cite{guarnieri:spectector, speculative-execution-combinations, binsec-haunted, ctfoundations:pldi:20, guanhua:kleespectre, guo:specusym} is the most widely used technique to detect Spectre vulnerabilities in programs.
However, existing detection tools do not scale well to large programs or to new speculation primitives.
For example, none can detect Spectre-\btb{} vulnerabilities due to limitations of \textit{always-mispredict} semantics~\cite{guarnieri:spectector, binsec-haunted, cauligi:sok-spectre-sw-defense}.
Other approaches include \textit{fuzzing}~\cite{oleksenko:specfuzz, oleksenko:revizor,nemati:scamv} or \textit{static analysis}~\cite{2021:oo7, lcms:mosier:2022}. 

\paragraph{\bf Formal Foundations}
Recent work deploys formal techniques to model and mitigate Spectre attacks~\cite{barthe:highassurance,ctfoundations:pldi:20,exorcisingspectrev1,guarnieri:spectector,vassena:blade,guanciale:inspectre,guarnieri:contracts,exorcisingspectrev1,shivakumar:spectredeclassified}.
%
%
Patrignani and Guarnieri~\cite{exorcisingspectrev1} and Shivakumar et al.~\cite{shivakumar:spectredeclassified} both use formal models to demonstrate that LLVM's \slh{} mitigation is incomplete for Spectre-\pht{}.
%
Fabian et al.~\cite{speculative-execution-combinations} define an extensible framework for composing semantics for individual speculation primitives to model and detect Spectre leakage due to their combinations. 
Mosier et al.~\cite{lcms:mosier:2022} and Ponce-de-León and Kinder~\cite{cat-spectre} take an axiomatic approach inspired by memory consistency models to model and detect program instructions that may transiently access or leak secrets.

\paragraph{\bf Software Mitigation}
Few compiler-based Spectre mitigation proposals~\cite{cauligi:sok-spectre-sw-defense} are formally grounded~\cite{vassena:blade} or protect against multiple Spectre variants~\cite{Narayan:swivel, venkman}.
Some detection tools~\cite{lcms:mosier:2022, cheang2019formal} can mitigate detected Spectre vulnerabilities but do not evaluate the performance of the resulting program.
None of these mitigations are readily deployable in an existing toolchain, in contrast with LLVM's \lfence{} and \slh{} mitigations (\S\ref{sec:mitigation-strategies}).
\paragraph{\bf Conclusions}
We present \tool{}, the first comprehensive mitigation for existing hardware that prevents all 
Spectre-\pht{}/\bti{}/\rsb{}/\stl{}/\psf{} leakage in \ourct{} programs.
We prove \tool{}'s correctness using our operational semantics, \ourmodel{}, and implement it as a code artifact, \toolclang{}, in the LLVM compiler infrastructure.
We evaluate \toolclang{} on a suite of cryptographic primitives from \libsodium{}, \hacl{}, and OpenSSL, and demonstrate significant performance and security improvements over the state-of-the-art. 


\section*{Acknowledgment}
This work was supported in part by the National Science Foundation (NSF) under award numbers CNS-2153936 and CAREER CCF-2236855; by a gift from Intel; and by the German Federal Ministry of Education and Research (BMBF) 
through funding for the CISPA-Stanford Center for Cybersecurity 
(FKZ: 13N1S0762). We thank the anonymous reviewers for their valuable feedback during the review process. We also thank Carlos Rozas and Jason Brandt from Intel for clarifying the technical implementation details of Intel CET-IBT.

\bibliographystyle{IEEEtran}
\bibliography{references}

\appendices

\section{Supplemental Material}

\subsection{Proof that \ourct{} is a Strengthening of \ctprog{}}
\label{app:sec:cts->ct}
\begin{theorem}
\label{app:thm:cts->ct}
    If a program $\Prog$ for \ourmodel{} satisfies \ourct{} (\Cref{def:ourct}), then it also satisfies \ctprog{} (\Cref{def:constant-time}).
\end{theorem}
\begin{proof}
    We prove the contrapositive.
    Suppose a \textit{sequential} trace $e$ contains a \ctprog{} violation, i.e., 
    it exposes a secretly-labeled sensitive operand $\rxmit$ of some transmitter $I_i$ via observation $O_i$.
    By \ourct{} rule \RuleTT{}.\ref{item:tt:xmit}, $\rxmit$ must be publicly-typed at $I_i$ (i.e., $\STReg{{F_i}}(\rxmit, I_i) = \Pub$).
    Thus, $e$ violates the public type of $\rxmit$ by sequentially assigning it a secret value, violating \ourct{} rule \RuleTT{}.\ref{item:tt:conservative}.
    Thus, $\Prog$ is not \ourct{}.
\end{proof}

\subsection{Proof of Theorem \ref*{lem:tcfg-complete}}
\label{app:sec:tcfg-complete}


\begin{proof}
    Let $I_i \in F$, and let $I_j \in F$ ($i < j$) be the \textit{next} executed instruction belonging to $F$ in some trace $e$ 
    (i.e., for all $l$ where $i \!<\! l \!<\! j$, $I_l \!\not \in\! F$).
    We will show that $I_i = I_j$ or $I_j \in \TSuccs(I_i)$ and thus $I_i \tcfgs{F} I_j$.
    The claim follows due to the transitivity of $\tcfgs{F}$.

    If $I_i$ halts execution (i.e., $C_i = C_{i+1}$) due to an $\Lfence$ or trap, 
    then clearly $I_j = I_i$, so we are done.

    Now, suppose $I_i \mapsto \MyCall \mid \Ret$.
    Then $I_i \in \mathbf{I_\mathrm{exit}}$
    and the instruction that executed before $I_j$ must have been a $I_{j-1} \mapsto \MyCall \mid \Ret$ instruction in order to have re-entered $F$ at step $j$.
    Thus, $I_j$ is a $\Endcall$ or post-$\MyCall$ instruction (due to \ourmodel{}'s control-flow restrictions, \S\ref{sec:model:call}--\ref{sec:model:ret}), and so $I_j \in \mathbf{I_\mathrm{enter}}$.
    Thus, $I_j \in \TSuccs(I_i)$, so we are done.

    Else, the instruction executing after $I_i$ is in the same procedure, or $I_{i+1} \in F$,
    so $I_{i+1} = I_j$.
    Since $I_{i+1} \in \Succs(I_i)$, $I_{i+1} = I_j \in \TSuccs(I_i)$  by definition (\S\ref{sec:tool:fence:tcfg}). 
\end{proof}

\subsection{Proof of Theorem \ref*{thm:source-sink}}
\label{app:sec:source-sink}


\begin{proof}
    Let $(I_i, I_k)$ be a source-sink pair, 
    and suppose for contradiction that $I_i$ executed transiently.
    By \Cref{lem:tcfg-complete}, there is a path $I_i \to_\mathrm{tcfg*}^F I_k$ in the T-CFG for $F$.
    Furthermore, since $I_i$ did not trap, the path is non-trivial (i.e., for some $i \leq j < j' \leq k$, $I_j \tcfg{F} I_{j'}$).
    \tool{}'s Fence Insertion Pass inserted an $\Lfence$ along all paths in the T-CFG from $I_i$ to $I_k$ (\S\ref{sec:tool:fence:insertion}), 
    thus $I_i \tcfgs{F} I_j \tcfgs{F} I_k$ for some $I_j \mapsto \Lfence$ ($i < j < k$).
    Recall that $\Lfence$s halt transient execution (\S\ref{sec:model:semantics:lfence}), so $I_j = I_k$.
    Thus, $I_k$ is not a sink (i.e., the real sink did not execute), a contradiction.
\end{proof}


\subsection{Proof of Theorem \ref*{thm:fps-inbounds}}
\label{app:sec:fps-inbounds}


\begin{proof}
    It suffices to show that each private stack pointer (PSP) load $I_k \mapsto \Load\,[\ZR+\psp{F}],\SP$ inserted by the FPS Pass at procedure (re-)entrypoints (L\ref{line:fps:prolog:ld} or L\ref{line:fps:call:ld} in \Cref{fig:fps-example}) returns an \textit{in-bounds} and \textit{aligned} value for the current procedure $F$.

    We use the following predicate to capture whether $\SP$ is in-bounds and aligned at step $i$, where $I_i$ is in procedure $F$ (recall from \S\ref{sec:tool:fps-pass} that $\StackBase{F}$, $\StackEnd{F}$, and $\FrameSize{F}$ denote $F$'s stack base, stack end, and frame size, respectively):
    $$Q(i) : \exists m \geq 0, \StackBase{F} \leq R_i(\SP) = \StackEnd{F} - m \cdot \FrameSize{F}.$$
    We show that the result of each PSP load $I_k$ satisfies $Q(k+1)$ by induction on the number of such loads. 

    \textit{Base case.}
    If $I_k$ is the first PSP load to execute, it executed in the prologue (L\ref{line:fps:prolog:ld}) of the first procedure, before any stores have executed.
    Thus, $I_k$ read from initial memory, which the FPS Pass initialized to $D_0(\psp{F}) = \StackEnd{F}$.
    Therefore, $Q(k+1)$ is satisfied for $m = 0$.

    \textit{Inductive step.}
    Suppose that for all PSP loads $I_i$ executing before $I_k$ ($i < k$), $Q(i+1)$ is satisfied (i.e., they return an in-bounds and aligned PSP for their respective procedures).
    
    If $I_k$ read from initial memory, $Q(k+1)$ holds for the same reason as in the base case.
    Otherwise, $I_k$ read from a prior same-address store $I_j$.
        
    Suppose for contradiction \textbf{$\bm{I_j}$ is an NCA store}.
    The FPS Pass only inserts a PSP load $I_k$ after a procedure (re)entrypoint, so some $\MyCall/\Ret$ executed between $I_j$ and $I_k$, forming a \ncasctrl{} source-sink pair. Thus, $I_j$ executed sequentially (\Cref{thm:source-sink}). 
    This implies $I_j$ in $\ProgIn$ sequentially wrote to 
    global address $\psp{F}$, 
    which was unmapped prior to the FPS Pass,
    contradicting \RuleWF{}.\ref{item:wf:segfault}.

    Suppose for contradiction \textbf{$\bm{I_j}$ is a CA stack store} in some procedure $G$.
    Since $I_k$ is a CA global load, the PSP load $I_i$ directly preceding $I_j$ read an out-of-bounds PSP that pointed into global memory. So, $R_{i+1}(\SP) \not \in \fps{G}$, 
    a contradiction of the inductive hypothesis
    that $Q(i+1)$ holds.

    Thus, \textbf{$\bm{I_j}$ is a CA global store}. Furthermore, it is the PSP store $I_j \mapsto \Store \, [\ZR+\psp{F}],\SP$ (L\ref{line:fps:prolog:st}, L\ref{line:fps:epilog:st}), since only instructions inserted by the FPS Pass access PSP global variables which were newly allocated in $\ProgFPS$.
    
    Let $I_i$ be the PSP load that directly preceded $I_j$ in $F$.
    If $I_j$ is in the prologue, then 
    $I_{j-3} \!=\! I_i$, $I_{j-2}$, $I_{j-1}$, and $I_j$ correspond to L\ref{line:fps:prolog:ld}--\ref{line:fps:prolog:st} in \Cref{fig:fps-example}.
    Applying the inductive hypothesis that $Q(i\!+\!1)$ holds, $R_{i+1}(\SP) \!=\! R_{j-2}(\SP) \!=\! \StackEnd{F} - m \cdot \FrameSize{F}$ for some $m \!\geq\! 0$.
    After frame allocation $I_{j-2}$, $R_{j-1}(\SP) \!=\! \StackEnd{F} - (m\!+\!1) \cdot \FrameSize{F}$.
    The bounds clip $I_{j-1}$ ensures $\StackBase{F} \leq R_j(\SP) \!=\! \StackEnd{F} - (m + d) \cdot \FrameSize{F}$ for a $d \in \{0,1\}$.
    Thus, $Q(j)$ holds.
    Since the PSP load $I_k$ read from the PSP store $I_j$, $Q(k+1)$ holds as well, and we are done.
    If $I_j$ is in the epilogue (L\ref{line:fps:epilog:st}), the argument is analogous.
\end{proof}

\subsection{Additional Instruction Semantics}
\label{app:semantics}
\noindent \semantics{
    &\frac{
        \begin{aligned}
            &\textsc{Arithmetic Op} \quad I \mapsto \Op_o \, r, \Vec{r_s} \quad
            \colornt v_l = o_\LabelSet(R(r_{s,1}), \ldots, R(r_{s,n})) \\
            &\colornt R'_\TokNT = R[\IncPC;\,r \update v_l] \quad
            \colornt C'_\TokNT = C[R \update R'_\TokNT]
        \end{aligned}
    }{
        \colornt \TransitionNT(C, P) = (C'_\TokNT, \ObsNone) \qquad
        \colort \bm{\TransitionT(C, P)} = \emptyset
    }
} \\
\noindent \semantics{
    &\frac{
        \begin{aligned}
            &\textsc{Uncond. Jump} \quad I \mapsto \Jmp \, d \quad
            \colornt R'_\TokNT = R[\PC \gets R(\PC) + 1 + d_\Pub]
        \end{aligned}
    }{
        \colornt \TransitionNT(C, P) = (C[R \update R'_\TokNT], \ObsNone) \qquad
        \colort \bm{\TransitionT(C, P)} = \emptyset
    }
}
\subsection{Extending \ourmodel{} and \tool{}}
\label{sec:extension}

\subsubsection{Extending the Execution Model}
\tool{} reasons directly about taint primitives.
Thus, when a new speculation primitive is discovered, one must determine if/how it changes the set of \textit{taint primitives}.

\paragraph{\bf Adding taint primitives}
\label{app:model:straight}
Consider straight-line speculation (SLS)~\cite{arm:straight-line-speculation}, which adds the following \textit{transient} transition to $\MyCall/\Ret/\Jmp$: $(C[\PC\texttt{++}], O) \in \TransitionT(C,P)$, where $O$ is the exposed observation.
One can prove that these transitions add a fifth taint primitive to \ourct{} programs (\S\ref{sec:ourprimitives}): \textbf{LINE}: a transient $\MyCall/\Ret/\Jmp$ that falls through to the next linear address (any register).
\tool{} can be extended to mitigate LINE with a new pass that inserts an $\Lfence$ following each $\Jmp$.
$\MyCall$s/$\Ret$s require no additional mitigations due to TYP.\ref{item:tt:arg} (\S\ref{sec:security-typeable}) and Register Cleaning (\S\ref{sec:tool:arg-pass}).
\Cref{fig:benchmark-results} evaluates this extension, \toolsls{}.

\paragraph{\bf Replacing taint primitives}
\label{app:model:psf}
Consider \psf{} (\S\ref{sec:data-flow-prediction}), which redefines the \textit{transient} transitions of $\Load$ as follows:
\\
\semantics{
    &\frac{
        \begin{aligned}
            &\textsc{Load (with PSF)} \quad \Load\,[r_a+d],r_v \qquad
            \colorspec A_l = R(r_a) + d_\Pub \\
            &\colort \mathbf{v_\TokT} = \{ D({\colorspec A})\} \cup \{v_{l'} \mid \exists A', (A', v_{l'}) \in S \} \setminus {\colornt v_\text{seq}} \quad
            \colorspec O = \ObsLoad{A_l}
        \end{aligned}
    }{
        \begin{aligned}
            \colort \bm{\TransitionT(C, P)} = \left\{
                (C[R \gets R[\IncPC; r \gets v_t] ;\, T \gets \T], {\colorspec O}) \mid v_t \in \mathbf{v_\TokT}
            \right\}
        \end{aligned}
    }
} \\
One can prove that these new transitions replace NCAL, NCAS, and STKL (\S\ref{sec:properties:ourprimitives}) with one new taint primitive:
\textbf{LOAD}: a transient $\Load$ (output register).
We evaluate an extension of \tool{} which mitigates Spectre-\psf{} in software using the prior rule in \S\ref{sec:case-study} and \S\ref{sec:results}.

\subsubsection{Extending the Leakage Model}
On many processors, $\mathtt{DIV}$ is a transmitter.
To model this, we can add ``$\mathtt{div}~a, b$'' to \ourmodel{}'s observation set $\ObsSet$ (\S\ref{sec:leakage-model}) and define the sequential transition for $\mathtt{DIV}\,r_a,r_b$ to expose the observation ``$\mathtt{div}~R(r_a), R(r_b)$.''
Adding new observations does not change \tool{}'s correctness proof.

\subsection{\tool{}'s Heuristic Directed Multicut}
\label{app:sec:multicut}
Computing the minimum directed multicut (i.e., the best global cut for multiple source-sink pairs) of a directed graph is NP-hard~\cite{minimum-directed-multicut}.
We approximate the optimal multicut by iteratively computing the optimal cuts for each source-sink pair individually (computed using Ford-Fulkerson~\cite{ford-fulkerson}) while \textit{fixing} the cuts of other source-sink pairs until convergence or loop.
After that, we validate the cutset by verifying no source can reach its corresponding sink.


\subsection{\ourct{}-Unsafe Compiler Optimizations}
\label{app:unsafe-compiler}
When compiling with \ourcc{}, we disable these optimizations which violate \ourct{} (\S\ref{sec:propreties:ourct}). 
\texttt{-mllvm -no\allowbreak{}-stack\allowbreak{}-slot\allowbreak{}-sharing}:
    \textit{Stack slot sharing} may assign two stack variables of different security types to the same frame index if their lifetimes do not overlap, violating \ourct{}'s stack typing requirements (\S\ref{sec:security-typeable}).
%
\texttt{-mno\allowbreak{}-red\allowbreak{}-zone}:
    A leaf procedure can use a limited amount of memory directly below the stack pointer (the ``red zone'') for its stack frame without needing to allocate it, violating \RuleWF{}.\ref{item:wf:stack} (\S\ref{sec:well-formed}).
\texttt{-mllvm -no\allowbreak{}-argument\allowbreak{}-promotion}:
    LLVM may promote (possibly secret) pass-by-reference arguments to pass-by-value during interprocedural optimization, violating \RuleTT{}.\ref{item:tt:arg} (\S\ref{sec:security-typeable}).
\texttt{-fno\allowbreak{}-jump\allowbreak{}-tables}:
    \ourct{} procedures contain exactly one $\Endcall$, marking the procedure entrypoint (\S\ref{sec:properties:procedures});
    jump tables require inserting $\Endcall$s elsewhere.

\subsection{\ourcc{} Implementation Details}
\label{app:implementation-details}
We implement Fence Insertion as a post-optimization IR pass,
Function-Private Stacks as a post-register-allocation machine IR (MIR) pass that runs during frame lowering, 
and Register Cleaning as a post-register-allocation MIR pass that runs after call lowering.

It is safe to implement these passes at the identified points in the LLVM pipeline, given that LLVM upholds the following:
(1) no NCA loads/stores are inserted when lowering IR to MIR or machine code (MC); and
(2) MIR/MC optimizations are not allowed to reorder loads/stores past fences.
The LLVM documentation~\cite{llvm-code-generator} and our analysis of output machine code corroborate these assumptions.



\subsection{Intel CET-IBT Implementations}
\label{app:cet}
Intel processors implement two variants of speculative semantics for CET-IBT: \textit{strong} and \textit{weak} (our terminology).

\textbf{Strong CET-IBT} completely blocks speculation following a missing \texttt{ENDBRANCH} instruction.
Like prior work~\cite{Narayan:swivel}, \ourmodel{} and \tool{} assume a processor with strong CET-IBT.
Intel processors implementing strong CET-IBT include Alder Lake N and Arizona Beach.
Intel has shared with us that the long-term direction of CET-IBT implementations is towards these strong speculative semantics.

\textbf{Weak CET-IBT} allows some fixed, nonzero number of instructions to speculatively execute following a missing \texttt{ENDBRANCH}.
Weak CET-IBT implementations are characterized by the number of instructions, and specifically loads, which may speculatively execute following a missing \texttt{ENDBRANCH}.
\textit{Older} weak CET-IBT implementations---found in Tiger Lake---allow up to 7 speculative instructions containing up to 5 speculative loads.
\textit{Newer} weak CET-IBT implementations---found in Alder Lake (our workstation), Sapphire Rapids, Raptor Lake, and some future processors---only allow up to 2 speculative instructions containing up to 1 speculative load.
Our preliminary investigation shows that it is possible to restore \tool{}'s security guarantees on older and newer CET-IBT implementations with moderate and negligible runtime cost, respectively.
\newpage

\section{Meta-Review}

\subsection{Summary}
The paper proposes \tool{}, a set of compiler passes aimed at hardening cryptographic code against Spectre-type attack. The paper introduces a stronger notion of constant time programming, which code must adhere to, and relies on hardware support for preventing leaks of secret data from transient execution.

\subsection{Scientific Contributions}
\begin{itemize}
    \item Creates a new tool to enable future science.
    \item Addresses a long-known issue.
    \item Provides a valuable step forward in an established field.
\end{itemize}

\subsection{Reasons for Acceptance}
\begin{enumerate}
    \item The paper addresses a long-standing problem. It proposes a new approach for a solution. The solution is comprehensive, proven secure, and only incurs a modest performance overhead.
\end{enumerate}

\subsection{Noteworthy Concerns}
\begin{enumerate}
    \item The proposed execution model that underlies the security proof assumes mostly in-order execution.
    There is a gap in semantics between the model and out-of-order execution. 
    \item \tool{} does not handle declassification.
    \item \tool{} relies on partially documented processor features and LLVM properties. 
    \tool{} is fragile due to possible future changes that break its assumptions.
\end{enumerate}

\section{Response to the Meta-Review}

\paragraph{\bf Response to Concern 1}
While ASP has an in-order semantics like many prior speculative processor models~\cite{guarnieri:spectector, speculative-execution-combinations, exorcisingspectrev1, shivakumar:spectredeclassified}, it is intended to capture \textit{all transient control- and data-flows} that may be exhibited by a program running on a speculative out-of-order processor.
Specifically, we believe our transient load and store semantics capture all Spectre-relevant memory instruction reorderings that can occur on a speculative out-of-order processor.

\paragraph{\bf Response to Concern 2}
As presented, \tool{}, like prior work \cite{vassena:blade}, does not handle secret declassification in its formalization or proof. We expect that secret declassification can be achieved without compromising \tool{}’ guarantees by calling an auxiliary function that (1) copies a secret input buffer to a public output buffer and then (2) executes a speculation fence. Proving this would require changes to \ourct{} and \ourmodel{}, so we leave this to future work.

\paragraph{\bf Response to Concern 3}
Indeed, \tool{} relies on processor and compiler properties that are not \textit{guaranteed} to hold in the future.
However, we believe academic work like \tool{} is critical towards incentivizing hardware vendors and compiler writers to design processors and compilers that can support efficient Spectre mitigations.

\end{document}